\documentclass[11pt]{article}
\usepackage{graphicx}

	\usepackage{amssymb}
	\usepackage{amsfonts,amsmath}
	\usepackage{todonotes}
	\usepackage{subcaption}
	\usepackage[colorlinks, linkcolor=blue, citecolor=blue]{hyperref}
	\usepackage{threeparttable}
	\usepackage{microtype}
	\usepackage{pdflscape}
	\usepackage{longtable}
	\usepackage{soul}
	\usepackage{xcolor}
	
	\usepackage{authblk}
	\usepackage{amsthm}
	\usepackage{fullpage}
	
	\newcommand{\xa}{{a_1}}
	\newcommand{\xaa}{{a_2}}
	\newcommand{\xb}{{b_1}}
	\newcommand{\xbb}{{b_2}}
	\newcommand{\xc}{{c_1}}
	\newcommand{\xcc}{{c_2}}
	\newcommand{\xd}{{d_1}}
	\newcommand{\xdd}{{d_2}}

	\newcommand{\PoA}{\frac{s^3+s^2+s+1}{s^2+s+1}}
	\newcommand{\mi}{\scalebox{0.7}[1.0]{-}}
	\newcommand{\ma}{.}
	\newcommand{\zeroone}{*}
	
	\newcommand{\UBone}{\frac{s^3+s^2+s+1}{s^3+2}}
	\newcommand{\UBtwo}{\frac{s^2+2s+1}{2s+1}}
	\newcommand{\UBthree}{\frac{s+1}{s}}
	\newcommand{\UBfour}{\frac{s^3-s^2+2s-1}{s^3-s^2+s-1}}
	\newcommand{\UBfive}{\frac{s+1}{2}}
	\newcommand{\UBsix}{\frac{s^2-s+1}{s^2-s}}
	\newcommand{\UBseven}{\frac{s^2}{2s-1}}
	\newcommand{\UBeight}{\UBthree}

	\newtheorem{theorem}{Theorem}
	\newtheorem{lemma}[theorem]{Lemma}
	\newtheorem{proposition}[theorem]{Proposition}

	{\bfseries}{\itshape}
	\newtheorem{remark}{Remark}{\bfseries}{\itshape}
	\newtheorem{example}{Example}{\bfseries}{\itshape}
	\newtheorem{claim}{Claim}{\bfseries}{\itshape}

\definecolor{emphcol}{gray}{.75}
\newcommand{\take}[1]{\colorbox{emphcol}{$#1$}}

\begin{document}
\title{Selfish Jobs with Favorite Machines:\\
		 Price of Anarchy vs  Strong Price of Anarchy}
\author[1,2]{Cong Chen}%
\author[2]{Paolo Penna}
\author[1]{Yinfeng Xu}

%
\affil[1]{
School of Management, Xi'an Jiaotong University, Xi'an, China
}
\affil[2]{Department of Computer Science, ETH Zurich, Zurich, Switzerland
}
\maketitle

\begin{abstract}
	We consider the well-studied game-theoretic version of machine scheduling in which jobs correspond to \emph{self-interested} users  and machines correspond to resources. Here each user chooses a machine trying to minimize \emph{her own} cost, and such selfish behavior typically results in some \emph{equilibrium} which is not globally \emph{optimal}: An equilibrium is an allocation where no user can reduce her own cost by moving to another machine, which in general need not minimize the makespan, i.e., the maximum load over the machines. 
	
	We provide \emph{tight} bounds on two well-studied notions in algorithmic game theory, namely, the \emph{price of anarchy} and the \emph{strong price of anarchy} on machine scheduling setting which lies in between the related and the unrelated machine case. Both notions study the social cost (makespan) of the \emph{worst} equilibrium compared to the optimum, with the strong price of anarchy  restricting to a stronger form of equilibria. Our results extend a prior study comparing the price of anarchy to the strong price of anarchy for two related machines (Epstein \cite{Epstein2010}, Acta Informatica~2010), 
	thus providing further insights on the relation between these concepts. Our exact bounds give a qualitative and quantitative comparison between the two models. The bounds also show that the setting is indeed easier than the two unrelated machines: In the latter, the strong price of anarchy is $2$, while in ours it is strictly smaller.
\end{abstract}


\section{Introduction}
Scheduling jobs on \emph{unrelated} machines is a classical optimization problem. 
In this problem, each job has a (possibly different) processing time on each of the $m$ machines, and a schedule is simply an assignment of jobs to machines. For any such schedule, the load of a machine is the sum of all processing times of the jobs assigned to that machine. The objective is to find a schedule minimizing the \emph{makespan}, that is, the maximum load among the machines. 

In its \emph{game-theoretic} version, jobs correspond to \emph{self-interested} users who choose which machine to use accordingly without any centralized control, and naturally aim at minimizing their \emph{own cost} (i.e.~the load of the machine they choose).
This will result in some \emph{equilibrium} in which no player has an incentive to deviate, though the resulting schedule is not necessarily the optimal in terms of makespan. Indeed, even for two unrelated machines it is quite easy to find equilibria whose makespan is arbitrarily larger than the optimum. 

\begin{example}[bad equilibrium for two unrelated machines]\label{ex:bad-Nash}
	Consider two jobs  and two unrelated machines, where the processing times are given by the following table:
	\[\begin{array}{c|c|c}\hline
	& \text{job~1} & \text{job~2} \\\hline
	\text{machine~1}	& 1 & \take{s} \\ \hline
	\text{machine~2}	& 	\take{s} & 1 \\ \hline
	\end{array}
	\]
	The allocation represented by the gray box is a (pure Nash) equilibrium: if a job moves to the other machine, its own cost increases from $s$ to $s +1$. As the optimal makespan is $1$ (swap the allocation), even for two machines the ratio between the cost of the worst equilibrium and the optimum is unbounded (at least $s$).
\end{example}

The inefficiency of equilibria in games is a central concept in algorithmic game theory, as it quantifies the\emph{ efficiency loss} resulting from a \emph{selfish behavior} of the players. In particular, the following two notions received quite a lot of attention:
\begin{itemize}
	\item \emph{Price of Anarchy (PoA) \cite{koutsoupias1999worst}.} The price of anarchy is the ratio between cost of the \emph{worst} Nash equilibrium (NE) and the \emph{optimum}. 
	\item \emph{Strong Price of Anarchy (SPoA) \cite{andelman2009strong}.} The strong price of anarchy is the ratio between cost of the \emph{worst} strong Nash equilibrium (SE) with the \emph{optimum}. 
\end{itemize}
The only difference between the two notions is in the equilibrium concept: While in a Nash equilibrium no player can \emph{unilaterally} improve by deviating, in \emph{strong} Nash equilibrium no \emph{group} of players can deviate and, in this way,  all of them improve \cite{Aumann1959}.  For instance, the allocation in Example~\ref{ex:bad-Nash} is \emph{not} a SE because the two players could change strategy and both will improve.

Several works pointed out that the price of anarchy may be too pessimistic because, even for two unrelated machines, the price of anarchy is \emph{unbounded} (see Example~\ref{ex:bad-Nash} above). Research thus focused on providing bounds for the strong price of anarchy and comparing the two bounds according to the problem restriction: \vspace{-.4cm}
\begin{table}[h!]
	\centering
\begin{tabular}{cc|c|c}\hline
 & 	& $SPoA$ & $PoA$ \\\hline 
unrelated&	& $m$ & $\infty$  \\
& $m=2$ & $2$ & $\infty$ \\\hline
related	& & $\Theta(\frac{\log m}{(\log\log m)^2})$ &  $\Theta(\frac{\log m}{\log\log m})$ \\
& $m=2$ & $\frac{\sqrt{5}+1}{2} \simeq 1.618$ & $\frac{\sqrt{5}+1}{2} \simeq 1.618$ \\
& $m=3$ &  & $2$ \\\hline
identical	&  & $\frac{2m}{m+1}$ & $\frac{2m}{m+1}$ \\
& $m=2$ & $4/3$ & $4/3$ \\\hline
\end{tabular} 
\end{table}\vspace{-.4cm}

\noindent
In \emph{unrelated} machines, each job can have different processing times on different machines. In \emph{related} machines,  each job has a \emph{size} and each machine a \emph{speed}, and the processing time of a job on a machine is the size of the job divided by the speed of the machine. For \emph{identical} machines, the processing time of a job is the same on all machines. 
The main difference between the identical machines and the other cases is obviously that in the latter the processing times are different. For two related machines, the worst bound of $PoA$ and $SPoA$ is achieved only when the \emph{speed ratio} $s$ equals a specific value. Indeed, \cite{Epstein2010} characterize and compare the $PoA$ and the $SPoA$ for all values of $s$, showing that $SPoA<PoA$ only in a specific interval of values (see next section for details). The lower bound on the $PoA$ for two unrelated machines (Example~\ref{ex:bad-Nash}) is unbounded when the ratio between different processing time $s$ is unbounded. 

\subsection{Our contribution}
Following the approach by \cite{Epstein2010} on two related machines, we study the \emph{price of anarchy} and the \emph{strong price of anarchy} for the case of two machines though in a more general setting. Specifically, we consider the case of jobs with \emph{favorite machines} \cite{favorite} which is defined as follows (see Figure~\ref{fig:compare}).  
Each job has  a certain \emph{size}  and a \emph{favorite} machine; The processing time of a job on its favorite machine is just its \emph{size}, while on any non-favorite machine it is \emph{$s$ times slower}, where $s\geq 1$ is common parameter across all jobs. This parameter is the speed ratio when considering the special case of two related machines (see  Figure~\ref{fig:compare}). The model is also a restriction of unrelated machines and the bad NE in Example~\ref{ex:bad-Nash} corresponds to two jobs of size one in our model. That is, when $s$ in unbounded, the price of anarchy is \emph{unbounded} also in our model,
\begin{equation}\label{eq:bad:easy}
PoA \geq s \ .
\end{equation}
This motivates the study of strong equilibria and $SPoA$ in our setting. 
We provide exact bounds on both the $PoA$ and the $SPoA$ for all values of $s$. 
\begin{figure}[tb]
	\begin{minipage}{.48\textwidth}
		\centering
		\includegraphics[width=0.95\textwidth]{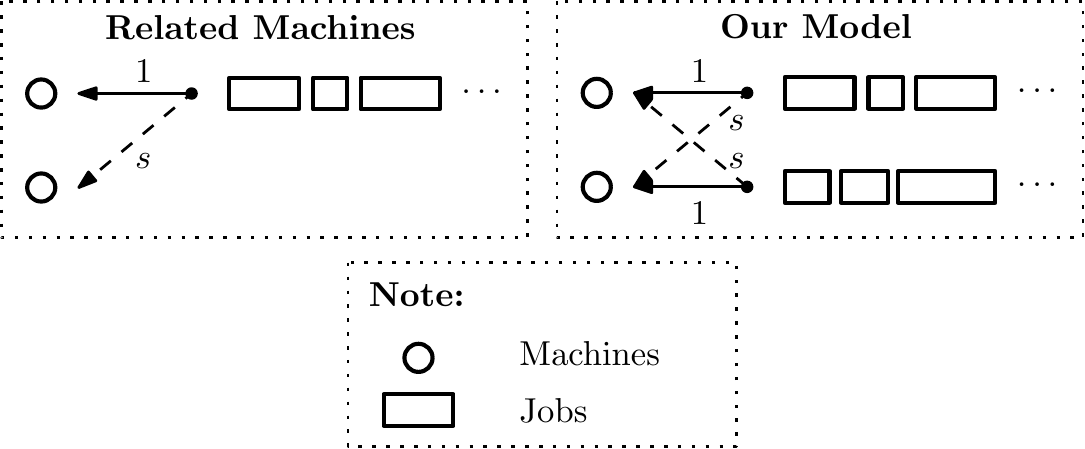}
		\caption{Relationship between 2 related machines and our model.}
		\label{fig:compare}
	\end{minipage}%
	\hfill
	\begin{minipage}{.48\textwidth}
		\centering
		\includegraphics[width=0.95\textwidth]{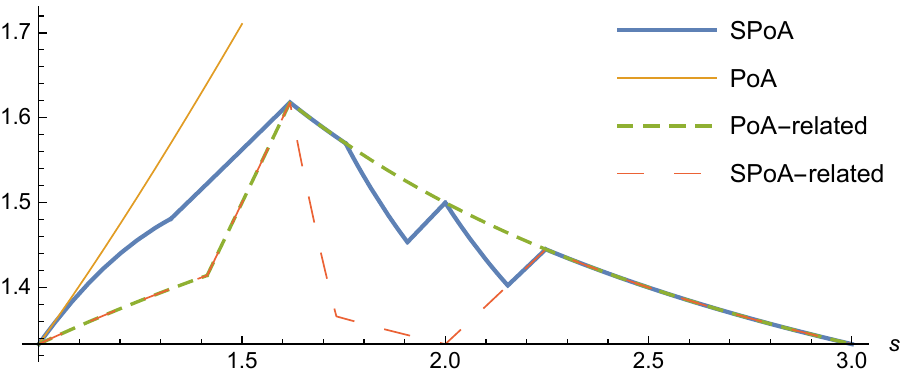}
		\caption{Comparison between $PoA$ and $SPoA$ in our model and in two related machines \cite{Epstein2010}.}
		\label{fig:bounds}
	\end{minipage}
\end{figure}

We first give an intuitive bound on $SPoA$  which holds for all possible values of $s$.
\begin{theorem}\label{thm.UB.simple}
	$SPoA \leq 1 + 1/s \leq 2$. 
\end{theorem}
By a more detailed an involved analysis, we prove further \emph{tight} bounds on $SPoA$. 
\begin{theorem}\label{thm.UB}
	$SPoA= \hat{\ell}_1$, where (see the blue line in Figure~\ref{fig:bounds})
	\[
	\hat{\ell}_1 = \begin{cases}
	\UBone{} \,, & 1 \:\,\le s \le s_1 \approx 1.325 \\
	\UBtwo{} \,, & s_1 \le s \le s_2 = \frac{1+\sqrt{5}}{2} \approx 1.618 \\
	\UBthree{} \,, & s_2 \le s \le s_3 \approx 1.755 \\
	\UBfour{} \,, & s_3 \le s \le s_4 \approx 1.907 \\
	\UBfive{} \,, & s_4 \le s \le s_5 = 2 \\
	\UBsix{} \,, & s_5 \le s \le s_6 \approx 2.154 \\
	\UBseven{} \,, & s_6 \le s \le s_7 \approx 2.247 \\
	\UBeight{} \,, & s_7 \le s \,.
	\end{cases}
	\]
\end{theorem}
At last we give \emph{exact} bounds on the $PoA$, which show that the bound in \eqref{eq:bad:easy} based on  Example~\ref{ex:bad-Nash} is never the worst case. 
\begin{theorem}\label{thm.poa}
	$PoA= \PoA{}$ (see the orange line in Figure~\ref{fig:bounds}).
\end{theorem}

These bounds express the dependency on the parameter $s$ and suggest a natural comparison with the case of two related machines (with the same $s$). 

\subsection{Related work}
The bad Nash equilibrium in Example~\ref{ex:bad-Nash} appears in several works \cite{andelman2009strong,originalSPOA,ChuKatPhrRot08,Gies17} to show that even for two machines the price of anarchy is unbounded, thus suggesting that the notion should be refined. Among these, the \emph{strong price of anarchy}, which considers \emph{strong} NE, is studied in \cite{andelman2009strong,fiat2007strong,Epstein2010}. The \emph{sequential price of anarchy}, which considers \emph{sequential} equilibria arising in extensive form games, is studied in \cite{originalSPOA,bilo,SPOAcongestion,Gies17}. In \cite{ChuKatPhrRot08} the authors investigate \emph{stochastically stable} equilibria and the resulting \emph{price of stochastic anarchy}, while \cite{KlePilTar09} focuses on the equilibria produced by the \emph{multiplicative weights update algorithm}. 
A further distinction  is between \emph{mixed} (randomized) and \emph{pure} (deterministic) equilibria: in the former, players choose a probability distribution over the strategies and regard their expected cost, in the latter they choose deterministically one strategy. In this work we focus on pure equilibria and in the remaining of this section we write \emph{mixed} $PoA$ to denote the bounds on the price of anarchy for mixed equilibria.

The following bounds have been obtained for scheduling games:

\begin{itemize}
\item \emph{Unrelated machines.} 
The $PoA$ is \emph{unbounded} even for \emph{two} machines, while the $SPoA$ is exactly $m$, for any number $m$ of machines \cite{fiat2007strong,andelman2009strong}. 

	\item \emph{Related machines.} 
		The price of anarchy is \emph{bounded} for \emph{constant number of machines}, and grows otherwise. Specifically, \emph{mixed} $PoA = \Theta(\frac{\log m}{\log\log\log m}) $ and $PoA =\Theta(\frac{\log m}{\log\log m})$ \cite{czumaj2007tight}, while  $SPoA = \Theta (\frac{\log m}{(\log\log m)^2})$ \cite{fiat2007strong}. The case of a \emph{small number of machines} is of particular interest. For \emph{two} and \emph{three} machines,  $PoA=\frac{\sqrt{5}+1}{2}$ and $PoA=2$ 
	 \cite{feldmann2003nashification}, respectively. For \emph{two} machines, exact bounds  as a function of the \emph{speed ratio} $s$ on both $PoA$ and $SPoA$ are given in \cite{Epstein2010}. 
	
	\item \emph{Restricted assignment.} The price of anarchy bounds are similar to related machines: \emph{mixed} $PoA = \Theta(\frac{\log m}{\log\log\log m}) $ and $PoA =\Theta(\frac{\log m}{\log\log m})$  \cite{AWERBUCH2006200}, where the analysis of $PoA$ is also in \cite{gairing2006price}. 
	\item \emph{Identical machines.} The price of anarchy and the strong price of anarchy for pure equilibria are \emph{identical} and bounded by a \emph{constant}: $PoA=SPoA=\frac{2m}{m+1}$ \cite{andelman2009strong}, where the upper and lower bounds on $PoA$ can be deduced from \cite{finn1979linear} and \cite{schuurman2007performance}, respectively. Finally, \emph{mixed} $PoA =\Theta(\frac{\log m}{\log\log m})$   \cite{koutsoupias1999worst,czumaj2007tight,koutsoupias2003approximate}.
	
\end{itemize}
For further results on other problems and variants of these equilibrium concepts we refer the reader to e.g.  \cite{roughgarden2009intrinsic,christodoulou2005price,PoSoriginal,Chien2009} and references therein.

\section{Preliminaries} 
\label{sec:setting}

\subsection{Model (favorite machines) and basic definitions}

In unrelated machine scheduling, there are $m$ machines and $n$ jobs. Each job $j$ has some processing time $p_{ij}$ on machine $i$. A schedule is an assignment of each job to some machine. The \emph{load} of a machine $i$ is the sum of the processing times of the jobs assigned to machine $i$. The \emph{makespan} is the maximum load over all machines. In this work, we consider the restriction of jobs with \emph{favorite} machines: Each job $j$ consists of a pair $(s_j,f_j)$, where $s_j$ is the \emph{size} of job $j$ and $f_j$ is the \emph{favorite machine} of this job. For a \emph{common parameter} $s\geq 1$, the processing time of a job in a favorite machine is just its size ($p_{ij}=s_j$ if $i=f_j$), while on non-favorite machines is it \emph{$s$ times slower} ($p_{ij}=s \cdot s_j$ if $i\neq f_j$).

We consider jobs as \emph{players} whose cost is the load of the machine they choose: For an allocation $x=(x_1,\ldots,x_n)$, where $x_j$ denotes the machine chosen by job $j$, and $\ell_i(x) = \sum_{j:x_j=i} p_{ij}$ is the load of machine $i$. We say that $x$ is a \emph{Nash equilibrium (NE)} if no player $j$ can unilaterally deviate and improve her own cost, i.e., move to a machine $\hat x_j$ such that $\ell_{\hat x_j}(\hat x)<\ell_{x_j}(x)$ where $\hat x = (x_1,\ldots, \hat x_j,\ldots, x_n)$ is the allocation resulting from $j$'s move. In a \emph{strong Nash equilibrium (SE)},
we require that in any group of deviating players, at least one of them does not improve: allocation $x$ is a SE if, for any $\hat x$ which differ in exactly a subset $J$ of players, there is one $j\in J$ such that  $\ell_{\hat x_j}(\hat x)\geq \ell_{x_j}(x)$. 
The \emph{price of anarchy (PoA)} is the worst-case ratio between the cost (i.e., makespan) of a NE and the optimum: $PoA = \max_{x\in \mathsf{NE}}\frac{C(x)}{opt}$ where $C(x)=\max_i \ell_i(x)$ and $\mathsf{NE}$ is the set of pure Nash equilibria. The \emph{strong price of anarchy (SPoA)} is defined analogously w.r.t. the set $\mathsf{SE}$ of strong Nash equilibria:  $SPoA = \max_{x\in \mathsf{SE}}\frac{C(x)}{opt}$. 

\subsection{First step of the analysis (reducing to eight groups of jobs)} To bound the strong price of anarchy we have to compare the \emph{worst SE} with the optimum.
The analysis consists of two main parts. We first 
consider the subset of jobs that needs to be reallocated in any allocation (or any SE) in order to obtain the optimum. It turns out that there are eight such subsets of jobs, and essentially the analysis reduced to the cases of eight jobs only. We then exploit the condition that possible reshuffling of these eight subsets must guarantee in order to be a SE.
\begin{figure}[b]
	\centering
	\includegraphics[width=0.8\textwidth]{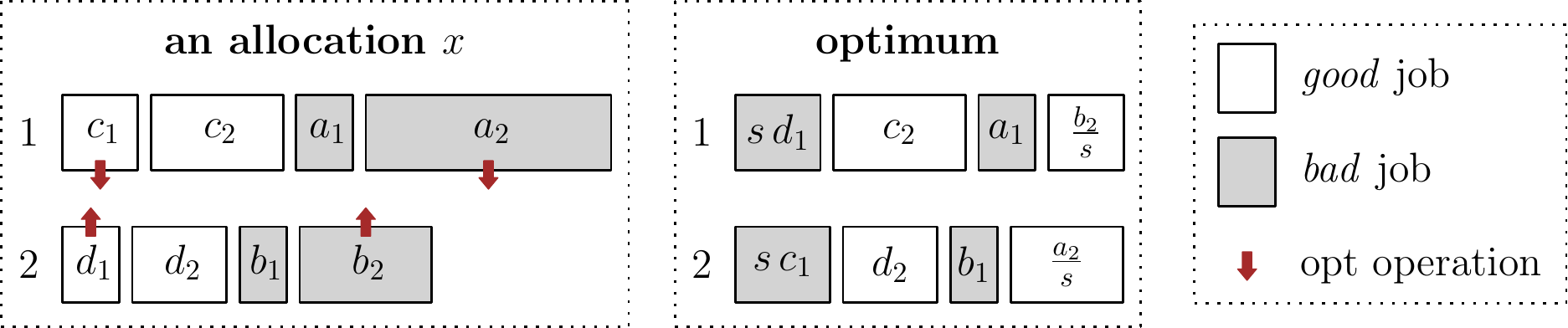}
	\caption{Comparing SE with the optimum.}
	\label{fig:setting}
\end{figure}

We say that a job is \emph{good}, for an allocation under consideration, if it is allocated to its \emph{favorite} machine. Otherwise the job is \emph{bad}. Consider an allocation $x$ and the optimum, respectively.
As shown in  Figure~\ref{fig:setting},
let 
\[
	\ell_1 = \xa{} + \xaa{} + \xc{} + \xcc{} \quad \text{and} \quad \ell_2 = \xb{} + \xbb{} + \xd{} + \xdd{}
\] 
be the load of machine 1 and machine 2 in allocation $x$,
where $\xa{}+ \xaa{}$ and $\xb{}+ \xbb{}$ are load of \emph{bad} jobs and $\xc{}+ \xcc{}$ and $\xd{}+ \xdd{}$ are load of \emph{good} jobs. Note that these quantities correspond to (possibly empty) subsets of jobs.
Without loss of generality suppose $\ell_1 \ge \ell_2$, that is
\[
C(x) = \max(\ell_1,\ell_2) = \ell_1 \ .
\]
In the optimum, some of the jobs will be processed by different machine as in allocation $x$. 
Suppose the jobs associated with $\xaa{},\xbb{},\xc{}$ and $\xd{}$ are the difference.
Thus, the load of the two machines in the optimum are 
\[
	\ell_1^* = \xa{}+ {\xbb{}}/{s} + \xcc{} + s\cdot \xd{} \quad \text{and} \quad
	 \ell_2^* = {\xaa{}}/{s} + \xb{} + s\cdot \xc{} + \xdd{} \,.
\]

\begin{remark}
 This setting can be also used to analyze $PoA$ and $SPoA$ for \emph{two related machines} \cite{Epstein2010}, which corresponds to two special cases $\xa{},\xaa{},\xd{},\xdd{} = 0$ and $\xb{},\xbb{},\xc{},\xcc{} = 0$.
\end{remark}

\subsection{Conditions for SE} 
\label{sec:analysis}
Without loss of generality we suppose $opt = 1$, i.e.,
\begin{gather}
	\label{eq.ll1}
	\ell_1^* \le 1 \, , \\
	\label{eq.ll2}
	\ell_2^* \le 1 \, .
\end{gather}
Since allocation $x$ is SE, we have that the minimum load in allocation $x$ must be at most $opt$, 
\begin{equation}\label{eq.l2}
	\ell_2 \le 1 \ ,
\end{equation}
because otherwise we could swap some of the jobs to obtain $opt$ and in this will improve the cost of all jobs.

Next, we shall provide several necessary conditions for allocation $x$ to be a SE.
Though these conditions are only a subset of those that SE must satisfy, they will lead to tight bounds on the $SPoA$.

\begin{enumerate}
	\item No job in machine 1 will go to machine 2:
		\begin{align}
			\label{eq.a1}
			& \ell_1 \le \ell_2 + {\xa{}}/{s}\,, && \hspace{-4cm} \text{for $\xa{}>0$}\,;\\
			\label{eq.a2}
			& \ell_1 \le \ell_2 + {\xaa{}}/{s}\,, && \hspace{-4cm} \text{for $\xaa{}>0$} \,; \\
			\label{eq.c1}
			& \ell_1 \le \ell_2 + s \cdot \xc{}\,, && \hspace{-4cm} \text{for $\xc{}>0$} \, ;\\
			\label{eq.c2}
			& \ell_1 \le \ell_2 + s \cdot \xcc{}\,, && \hspace{-4cm} \text{for $\xcc{}>0$} \, . 
	\intertext{\item No $\xaa{} \, \text{-} \, \xbb{}$ swap:}
			\label{eq.ab1}
			& \ell_1 \le \ell_2 -\xbb{} + {\xaa{}}/{s} \,,\\
			\label{eq.ab2}
			\text{or} \quad & \ell_2 \le \ell_1 -\xaa{} + {\xbb{}}/{s} \,.
	\intertext{\item No $\xaa{} \, \text{-} \, \xd{}$ swap:}
			\label{eq.ad1}
			& \ell_1 \le \ell_2 - \xd{} + {\xaa{}}/{s} \,,\\
			\label{eq.ad2}
			\text{or} \quad & \ell_2 \le \ell_1 -\xaa{} + s \cdot \xd{} \,.
	\intertext{\item No $\xaa{} \, \text{-} \, \{\xb{},\xdd{}\}$ swap:}
			\label{eq.abd1}
			& \ell_1 \le \ell_2 -\xb{} -\xdd{} + {\xaa{}}/{s}\,,\\
			\label{eq.abd2}
			\text{or} \quad & \ell_2 \le \ell_1 - \xaa{} + {\xb{}}/{s} + s \cdot \xdd{} \,.
		\intertext{\item No $\{\xaa{},\xcc{}\} \, \text{-} \, \{\xb{},\xbb{},\xd{},\xdd{}\}$ swap:}
			\label{eq.acbd1}
			& \ell_1 \le {\xaa{}}/{s} + \xcc{} \cdot s\,,\\
			\label{eq.acbd2}
			\text{or} \quad & \ell_2 \le \xa{} + \xc{} + (\xb{}+\xbb{})/{s}  + s (\xd{}+\xdd{}) \,.
	\intertext{\item No $\{\xa{},\xaa{},\xc{},\xcc{}\} \, \text{-} \, \{\xb{},\xbb{},\xd{},\xdd{}\}$ swap:}
			\label{eq.all1}
			& \ell_1 \le (\xa{}+\xaa{})/{s}  + s (\xc{}+\xcc{}) \,,\\
			\label{eq.all2}
			\text{or} \quad & \ell_2 \le (\xb{}+\xbb{})/{s}  + s (\xd{}+\xdd{})  \,.
		\end{align}
\end{enumerate}

\section{Strong Price of Anarchy} \label{sec:spoa}
We first prove a simpler  general upper bound (Theorem~\ref{thm.UB.simple}) and then refine the result by giving tight bounds for all possible values of $s$ (Theorem~\ref{thm.UB}). The first result says that the strong price of anarchy is bounded and actually gets better as $s$ increases.

\begin{proof}[Proof of Theorem \ref{thm.UB.simple}]
	We distinguish the following cases:
	\begin{description}
		\item[($a_2=0$.)] 
		By definition of $\ell_1$ and $\ell_1^*$ we have $\ell_1 \leq c_1 + \ell^*_1$, 
		and also by definition of $\ell_2^*$ we have $\ell_2^* \geq s \cdot c_1$.
		Therefore, along with \eqref{eq.ll1} and \eqref{eq.ll2} it holds that
		$\ell_1 \leq \frac{\ell^*_2}{s} + \ell^*_1 \leq \frac{1}{s} + 1$.
		
		\item[($a_2>0$.)] 
		In this case we use that at least one between \eqref{eq.ab1} or \eqref{eq.ab2} must hold. If \eqref{eq.ab1} holds then: 
		\begin{equation*}
		\ell_1 \le \ell_2 - b_2 + \frac{a_2}{s} \le d_1 + \ell_2^* \leq \frac{\ell^*_1}{s} + \ell^*_2 \leq \frac{1}{s}+ 1 
		\end{equation*}
		where the last two inequalities follow by the fact that $\ell_1^* \geq s \cdot d_1$ and \eqref{eq.ll1}-\eqref{eq.ll2}.
			
		If \eqref{eq.ab2} holds then we use that at least one between \eqref{eq.all1} or \eqref{eq.all2} must hold. If \eqref{eq.all1} holds then by definition of $\ell_1$ this can be rewritten as
		\[
		\frac{a_1+a_2}{s} \leq c_1 + c_2 \, .
		\]
		By adding $a_1+a_2$ on both sides, this is the same as
		\begin{equation*}
			a_1+a_2 + \frac{a_1+a_2}{s} \leq \underbrace{a_1 + a_2 + c_1 + c_2}_{\ell_1} \quad
			\Rightarrow \quad a_2 \leq  \frac{s}{s+1} \ell_1 \,.
		\end{equation*}
		By \eqref{eq.l2} and \eqref{eq.a2} we also have
		\[
		\ell_1 \leq 1 + \frac{a_2}{s}
		\]
		and by putting the last two inequalities together we obtain 
		\[
		\ell_1 \leq 1 + \frac{\ell_1}{s+1} \quad\Leftrightarrow \quad \ell_1 \leq  1 + \frac{1}{s} \, .
		\]
		If \eqref{eq.all2} holds then we first observe that, by definition, the following identity holds:
		\[
		\ell_1 = \ell^*_1 + \frac{\ell^*_2}{s} + \frac{s^2-1}{s^2}a_2 - \frac{b_1+b_2}{s} - sd_1 - \frac{d_2}{s} \, .
		\]
		We shall prove that 
		\begin{equation}\label{eq:rhs-negative}
		\frac{s^2-1}{s^2}a_2 - \frac{b_1+b_2}{s} - sd_1 - \frac{d_2}{s} \leq 0
		\end{equation}
		and thus conclude from \eqref{eq.ll1}-\eqref{eq.ll2} that
		\[
		\ell_1 \leq \ell^*_1 + \frac{\ell^*_2}{s} \leq 1 + \frac{1}{s} \, .
		\]
		From \eqref{eq.a2} and \eqref{eq.ab2} we have
		\[
		a_2 \leq \frac{b_2}{s-1} \leq \frac{b_1+b_2}{s-1}
		\]
		and plugging into left hand side of \eqref{eq:rhs-negative} we get 
		\[
		\frac{s^2-1}{s^2}a_2 - \frac{b_1+b_2}{s} - sd_1 - \frac{d_2}{s} \leq \frac{b_1+b_2}{s^2} - sd_1 - \frac{d_2}{s} \,.
		\]
		Finally, by definition of $\ell_2$, \eqref{eq.all2} can be rewritten as 
		\[
		\frac{b_1 + b_2}{s^2} \leq \frac{d_1 + d_2}{s}
		\]
		which implies 
		\[
		\frac{b_1 + b_2}{s^2} - sd_1 - \frac{d_2}{s} \leq \frac{d_1+d_2}{s} - s d_1 - \frac{d_2}{s} \leq 0
		\]
		which proves \eqref{eq:rhs-negative} and concludes the proof of this last case. 
	\end{description}
\end{proof}

\subsection{Notation used for the improved upper bound}
We shall break the proof into several subcases. First, we consider different intervals for $s$.
Then, for each interval, we consider the quantities $\xa{},\xaa{},\xc{},\xcc{}$ and break the proof into subcases, according to the fact that some of these quantities are  zero or strictly positive (Lemmas~\ref{lem.a1-1}-\ref{lem.c1-0c2-1} below). Finally, in each subcase, use a subset of the SE constraints to obtain the desired bound. 

Table~\ref{tab.UB} shows the subcases and which constraints are used to prove a corresponding bound. 
Note that for the chosen constraints, we also specify some \emph{weight} which essentially says how these constraints are combined together in the actual proof. 
We explain this with the following example. 

\subsubsection{An illustrative example (weighted combination of constraints)}\label{subsub.illustrative example}
Consider the case 1-2 in Table~\ref{tab.UB} (second row). In the third column, the four numbers show whether the four variable $\xa{},\xaa{},\xc{},\xcc{}$ are zero or non-zero, where ``$1$'' represents non-zero and ``$*$'' represents non-negative. The last column is the bound we obtain for $\ell_1$ and thus for the $SPoA$. Specifically, in this case  we want to prove the following:
\begin{claim}
	If $\xa{}>0$, $\xaa=0$, and $\xc{}>0$, then $\ell_1 \le \frac{s^3+s^2+s+1}{s^3+s^2+1}$.
\end{claim}
\begin{proof}
	First summing all the constraints with the corresponding weights given in columns 4 and 5 of Table~\ref{tab.UB} (second row):
	\[
	\textstyle
	 \eqref{eq.ll1}  \frac{s^3+s}{s^3+s^2+1} + \eqref{eq.ll2} \frac{s^2+1}{s^3+s^2+1} + \eqref{eq.a1} \frac{s^2}{s^3+s^2+1} + \eqref{eq.c1}  \frac{1}{s^3+s^2+1} \, .
	\]
	This simplifies  as
	\begin{equation}\label{eq.ill}
		\textstyle
	 \xa{} + \frac{s^3+s^2+s+1}{s^4+s^3+s}\xaa{} + \xc{} + \frac{s^3+s^2+s+1}{s^3+s^2+1}\xcc{} + \frac{s^4-1}{s^3+s^2+1}\xd{} \le \frac{s^3+s^2+s+1}{s^3+s^2+1}\, .
	\end{equation}
	Note that all the weights (column 5) are positive when $s$ in the given interval (column 2), so that the direction of the inequalities (column 4) remains.
	
	According to columns 2 and 3, we have $\xaa{} = 0$, $\frac{s^3+s^2+s+1}{s^3+s^2+1} \ge 1$ and $\frac{s^4-1}{s^3+s^2+1} \ge 0$. This and \eqref{eq.ill} imply 
	\begin{equation*}
	\ell_1 = \xa{} + \xaa{} + \xc{} + \xcc{} \le \frac{s^3+s^2+s+1}{s^3+s^2+1} \,.
	\end{equation*}
	We therefore get the bound in column 6.
\end{proof}
\subsection{The actual proof}
We break the proof of Theorem~\ref{thm.UB} into several lemmas, and prove the upper bound in each of them. 
The lemmas are organized depending on the value of $\xa{},\xaa{},\xc{},\xcc{}$.
Finally, we show that the bounds of these lemmas are tight (Lemma~\ref{lem.LB}).  

\begin{lemma}\label{lem.a1-1}
	If $\xa{} > 0$, then $\ell_1 \le \hat{\ell}_1$. 
\end{lemma}
\begin{lemma}\label{lem.a2-0}
	If $\xa{} = \xaa{} = 0$, then $\ell_1 \le \hat{\ell}_1$.
\end{lemma}
\begin{lemma}\label{lem.c1-1}
	If $\xa{} = 0$, $\xaa{} , \xc{} > 0$, then $\ell_1 \le \hat{\ell}_1$.
\end{lemma}
\begin{lemma}\label{lem.c1-0c2-0}
	If $\xa{} = \xc{} = \xcc{} = 0$ and $\xaa{} > 0$, then $\ell_1 \le \hat{\ell}_1$.
\end{lemma}
\begin{lemma}\label{lem.c1-0c2-1}
	If $\xa{} = \xc{} = 0$ and $\xaa{}, \xcc{} > 0$, then $\ell_1 \le \hat{\ell}_1$.
\end{lemma}
\begin{lemma}\label{lem.LB}
	The lower bound of $SPoA\geq \hat \ell_1$ is given by the instances  in \emph{Table~\ref{tab.LU}}.
\end{lemma}
Figure~\ref{fig:UB:subcases} illustrates the relation between the general bound and the bound proved in each of the these lemmas.

The proofs of these lemmas are based on Table~\ref{tab.UB}. Note that in Table~\ref{tab.UB}, each row has a bound for $\ell_1$ in the last column. Since the above illustrative example has already explained how the bounds are generated, here we mainly focus on the relationship of these bounds with the lemmas. 

The index (column 1) of each row in Table~\ref{tab.UB} encodes the relationship between those bounds (last column) and how these bounds should be combined to obtain the corresponding lemma.
Specifically, cases separated by ``\ma{}'' are subcases that should take \emph{maximum} of them due to we aim to measure the worst performance, while cases separated by ``\mi{}'' are a single case bounded by several different combinations of constrains that should take \emph{minimum} of them due to these constrains should hold at the same time. 
For instance, we consider three subcases in Lemma~\ref{lem.a2-0}, since $\xa{} = \xaa{} = 0$:
Case \ref{lem.a2-0}\ma{}1 ($\xc{}=0,\xcc{}>0$), Case~\ref{lem.a2-0}\ma{}2~($\xc{}>0,\xcc{}=0$) and Case~\ref{lem.a2-0}\ma{}3~($\xc{}>0,\xcc{}>0$), which correspond to rows \ref{lem.a2-0}\ma{}1 to \ref{lem.a2-0}\ma{}3\mi{}b in Table~\ref{tab.UB}.
The bound for Case~\ref{lem.a2-0}\ma{}3 is the minimum of the bounds of \ref{lem.a2-0}\ma{}3\mi{}a and \ref{lem.a2-0}\ma{}3\mi{}b.
Finally, the maximum of the bounds of the three subcases give the bound for Lemma~\ref{lem.a2-0}, i.e.,
$\max \big\{ 1,\, \frac{1}{s},\, \min\{\frac{2(s+1)}{s+2},\frac{s+2}{s+1}\} \big\}$ (orange line of Figure~\ref{fig:a2-0}).

The proofs of Lemmas~\ref{lem.a1-1}-\ref{lem.LB} are given in full detail in Appendixes~\ref{app:postponed proofs:lemmas} and \ref{app:proof_of_lower_bound}.
\begin{figure}[tb]
		\begin{subfigure}{.325\textwidth}
			\centering
			\includegraphics[width=1\textwidth,height=1.7cm]{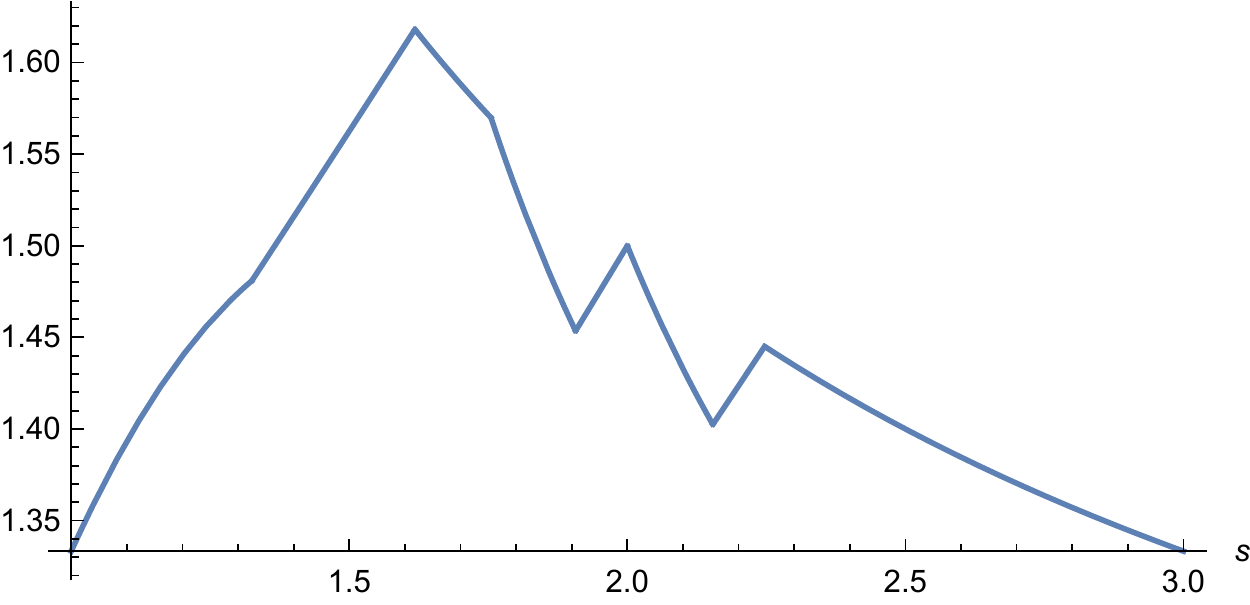}
			\caption{Theorem~\ref{thm.UB}.}
			\label{fig:UB}
		\end{subfigure}
		\begin{subfigure}{.325\textwidth}
			\centering
			\includegraphics[width=1\textwidth,height=1.7cm]{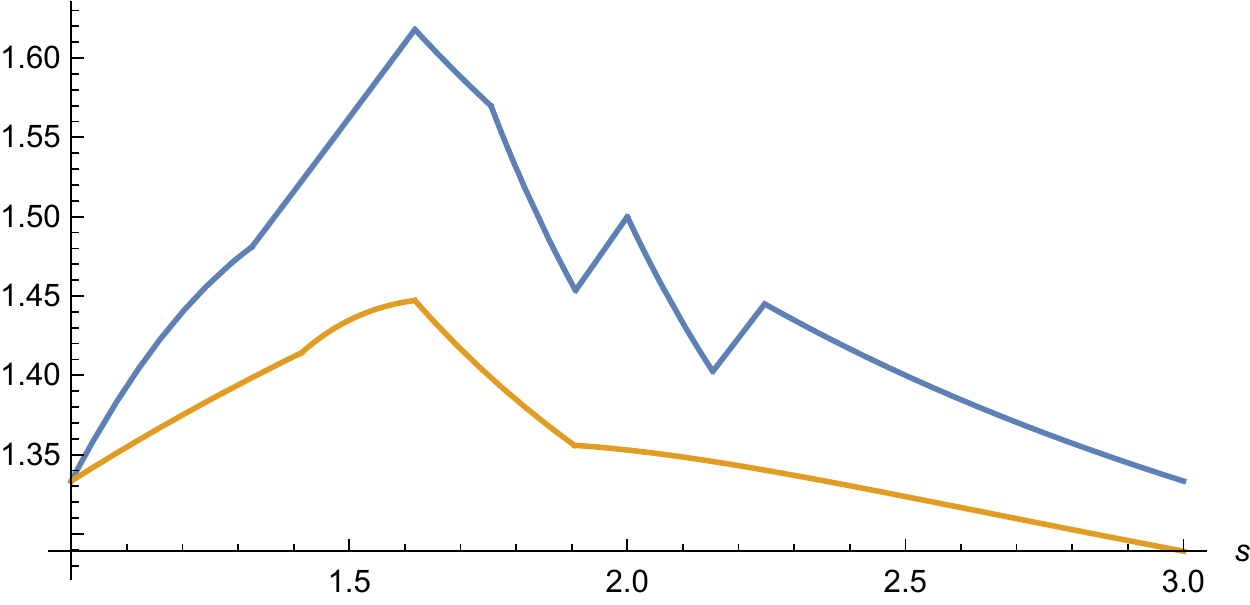}
			\caption{Lemma~\ref{lem.a1-1}.}
			\label{fig:a1-1}
		\end{subfigure}
		\begin{subfigure}{.325\textwidth}
			\centering
			\includegraphics[width=1\textwidth,height=1.7cm]{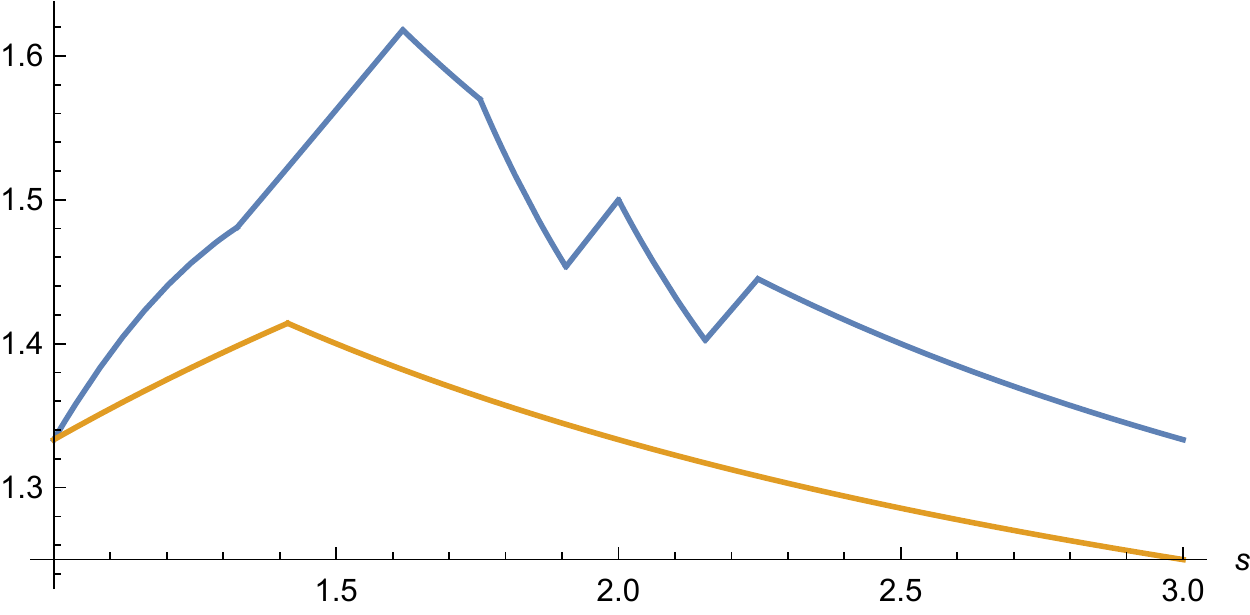}
			\caption{Lemma~\ref{lem.a2-0}.}
			\label{fig:a2-0}
		\end{subfigure}
		\begin{subfigure}{.325\textwidth}
			\centering
			\includegraphics[width=1\textwidth,height=1.7cm]{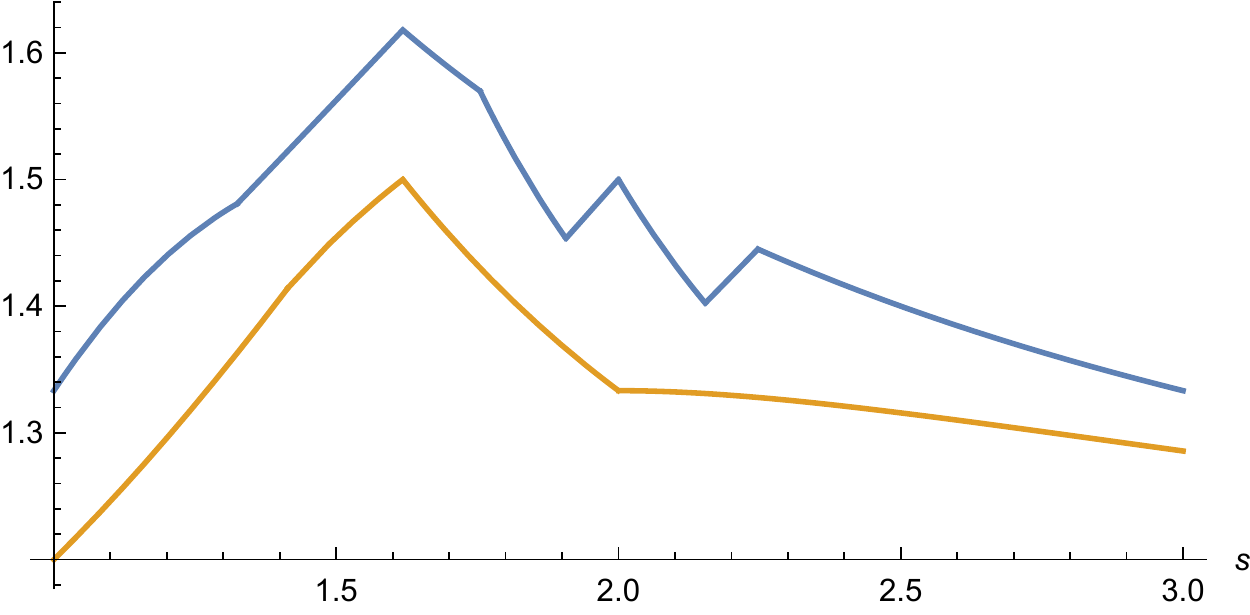}
			\caption{Lemma~\ref{lem.c1-1}.}
			\label{fig:c1-1}
		\end{subfigure}
		\begin{subfigure}{.325\textwidth}
			\centering
			\includegraphics[width=1\textwidth,height=1.7cm]{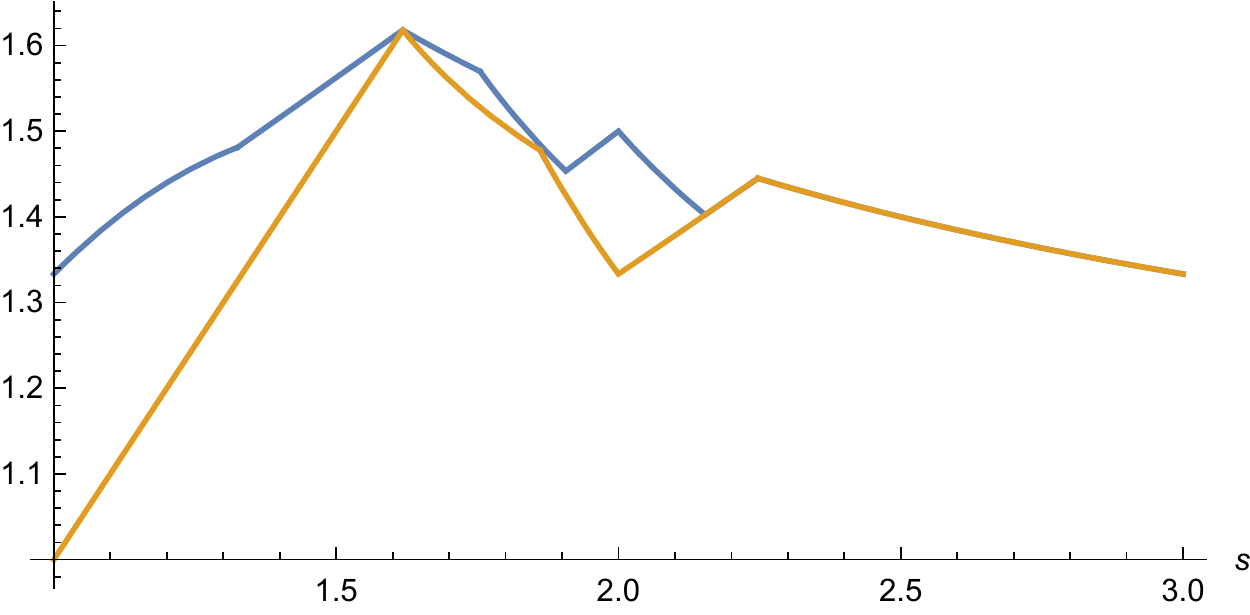}
			\caption{Lemma~\ref{lem.c1-0c2-0}.}
			\label{fig:c1-0c2-0}
		\end{subfigure}
		\begin{subfigure}{.325\textwidth}
			\centering
			\includegraphics[width=1\textwidth,height=1.7cm]{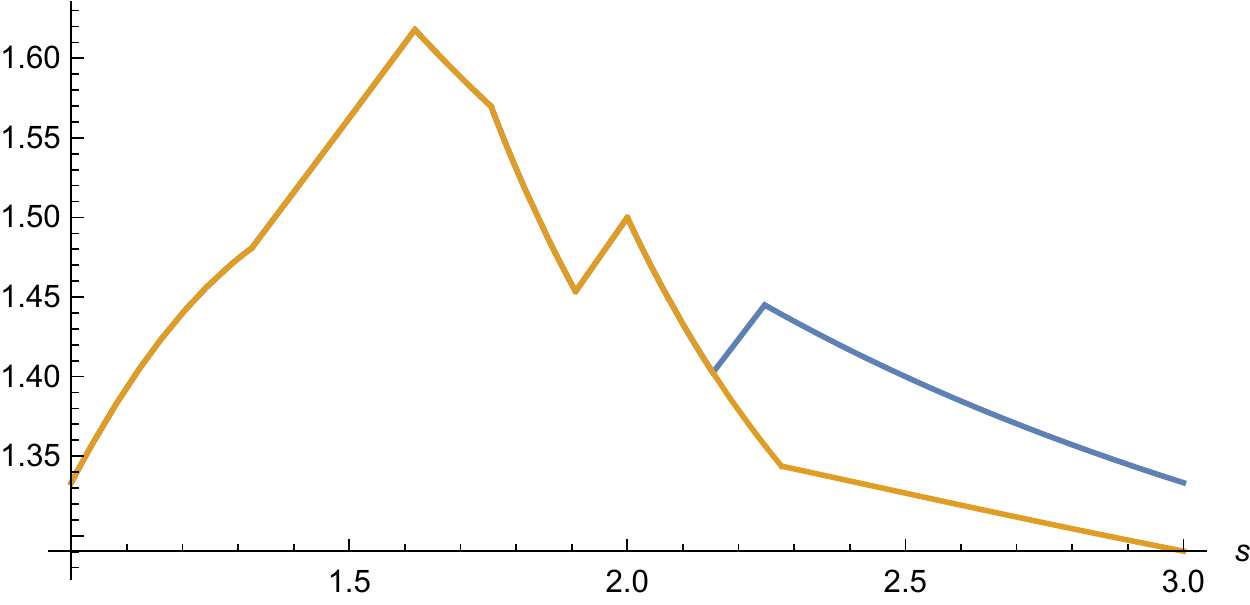}
			\caption{Lemma~\ref{lem.c1-0c2-1}.}
			\label{fig:c1-0c2-1}
		\end{subfigure}
		\caption{Proof of Theorem~\ref{thm.UB}.}
		\label{fig:UB:subcases}
	\end{figure}
\begin{table}[tb]
	\caption{Lower bound for $SPoA$.}
	\label{tab.LU}
	\centering

	\begin{threeparttable}
	\begin{tabular}{cccccccccc}
		\hline
	
		\hline
		 & $s$ & {$\xaa{}$} & {$\xbb{}$}  & {$\xcc{}$} & {$\xd{}$} & {$\xdd{}$} & {LB ($=\ell_1$)} & $\ell_2$\\
		\hline
		LB1 & $[0 , s_1]$ & $\frac{s^3+s^2}{s^3+2}$ & $\frac{s}{s^3+2}$  & $\frac{s+1}{s^3+2}$ & $\frac{s^2-1}{s^3+2}$ & $\frac{s^3-s^2-s+2}{s^3+2}$ & $\UBone{}$  & $\frac{s^3+1}{s^3+2}$\\
		LB2 & $[s_1 , s_2]$ & $\frac{s^2+s}{2s+1}$ & $\frac{s^2}{2s+1}$  & $\frac{s+1}{2s+1}$ & $0$ & $\frac{s}{2s+1}$ & $\UBtwo{}$ & $\frac{s^2+s}{2s+1}$ \\
		LB3 & $[s_2 , s_3]$ & $1$ & $s-1$ & $\frac{1}{s}$ & $0$ & $2-s$ & $\UBthree{}$ & $1$ \\
		LB4 & $[s_3 , s_4]$ & $\frac{s^2/(s-1)}{s^2+1}$ & $\frac{s^2}{s^2+1}$ & $\frac{s^2-s+1}{s^2+1}$ & $0$ & $\frac{1}{s^2+1}$ & $\UBfour{}$ & $1$ \\
		LB5 & $[s_4 , s_5]$ & $\frac{s}{2}$ & $\frac{s}{2}$ & $\frac{1}{2}$ & $0$ & $\frac{2-s}{2}$ & $\UBfive{}$ & $1$ \\
		LB6 & $[s_5 , s_6]$ & $\frac{1}{s-1}$ & $1$ & $\frac{s-1}{s}$ & $0$ & $0$ & $\UBsix{}$ & $1$ \\
		LB7 & $[s_6 , s_7]$ & $\frac{s^2}{2s-1}$ & $0$ & $0$ & $\frac{(s-1)^2}{2s-1}$ & $\frac{s-1}{2s-1}$ & $\UBseven{}$ & $\frac{s^2-s}{2s-1}$ \\
		LB8 & $[s_7 , \infty]$ & $\frac{s+1}{s}$ & $0$ & $0$ & $\frac{1}{s}$ & $\frac{s^2-s-1}{s^2}$ & $\UBeight{}$ & $\frac{s^2-1}{s^2}$ \\
		\hline
	
		\hline
		\end{tabular}
		\begin{tablenotes}
		\footnotesize
		\item  $^*$Note that $\xa{}=\xb{}=\xc{}=0$, and $\xaa,\xbb{},\xcc{},\xd{},\xdd{}$ each represent a single job here.
		\end{tablenotes}
	\end{threeparttable}
\end{table}\vspace{-0.4cm} 
\section{Price of Anarchy} 
\label{sec:poa}

In this section we prove the bounds on the $PoA$ in Theorem~\ref{thm.poa}. 
Suppose the smallest jobs in $\xa{},\, \xaa{},\, \xc{},\, \xcc{}$ are $\xa{}',\, \xaa{}',\, \xc{}',\, \xcc{}'$ respectively. 
To guarantee NE, it must hold that no single job in machine~1 can improve by moving  to machine~2, so that
\eqref{eq.a1}, \eqref{eq.a2}, \eqref{eq.c1} and \eqref{eq.c2} are also true for NE since $\xa{}' \le \xa{},\, \xaa{}' \le \xaa{},\, \xc{}' \le \xc{} \text{ and } \xcc{}' \le \xcc{}$.
Like in the analysis of the $SPoA$, we assume without loss of generality that $opt=1$, and thus \eqref{eq.ll1} and \eqref{eq.ll2} hold. We use these six constraints to prove Theorem~\ref{thm.poa}.

It is easy to see that if at most one of $\xa{},\, \xaa{},\, \xc{},\, \xcc{}$ is nonzero, then $\ell_1 \le s$, thus here we only discuss the cases where at least two of them are nonzero.
Similar to the proof of $SPoA$ the proofs of these lemmas are based on last four rows of Table~\ref{tab.UB}.

\begin{lemma}\label{cla.1}
	If $\xaa{}=0$, then $\ell_1 \le \frac{s^3+s^2+s+1}{s^2+s+1}$.
\end{lemma}

\begin{lemma}\label{cla.2}
	If $\xaa{}>0$ and $\xcc{}=0$, then $\ell_1 \le \frac{s^3+s^2+s+1}{s^2+s+1}$.
\end{lemma}

\begin{lemma}\label{cla.3}
	If $\xaa{}>0$ and $\xcc{}>0$, then $\ell_1 \le \frac{s^3+s^2+s+1}{s^2+s+1}$.
\end{lemma}
\begin{lemma}\label{cla.4}
	The lower bound of $PoA$ is achieved by the following case,
	\[
		\textstyle {\xaa{}} = \frac{s^3 + s^2}{s^2+s+1},\quad {\xb{}} = \frac{1}{s^2+s+1},\quad {\xbb{}} = \frac{s^3}{s^2+s+1},\quad {\xcc{}} = \frac{s+1}{s^2+s+1} \,,
	\]
	and $\xa{} = \xc{} = \xd{} = \xdd{} = 0 $.
\end{lemma}

Lemmas~\ref{cla.1}-\ref{cla.4} (proofs see Appendix~\ref{app:proofs_of_PoA}) complete the proof of Theorem~\ref{thm.poa}.

\section{Conclusion and open questions}
In this work, we have analyzed both the \emph{price of anarchy} and the \emph{strong price of anarchy} on a simple though natural model of two machines in which each job has its own \emph{favorite} machine, and the other machine is \emph{$s$ times slower} machine. The model and the results extend the case of two \emph{related machines} with speed ratio $s$ \cite{Epstein2010}. In particular, we provide \emph{exact} bounds on $PoA$ and $SPoA$ for \emph{all values of $s$}. On the one hand, this allows us to compare with the same bounds for two related machines (see Figure~\ref{fig:bounds}). On the other hand, to the best of our knowledge, this is one of the first studies which considers in the analysis the \emph{processing time ratio} between different machines (with the exception of \cite{Epstein2010}). Prior work mainly focused on the asymptotic on the number of machines (resources) or/and number of jobs (users). Instead, the loss of efficiency due to selfish behavior is perhaps also caused by the presence of \emph{different} resources, even when the latter are few.

Unlike for two related machines, in our setting the $PoA$ grows with $s$ and thus the influence of coalitions and the resulting $SPoA$ is more evident. 
Note for example that the $SPoA \leq \phi = \frac{\sqrt{5}+1}{2} \simeq 1.618$ and this bound is attained for $s=\phi$ exactly like for two related machines (see Figure~\ref{fig:compare}). Also, for sufficiently large $s$, the two problems have  the exact same $SPoA$, though the $PoA$ is very much different. 

It is natural to study the $PoA$ and $SPoA$ depending on the specific speed ratio, or processing time ratio. In that sense, it would be interesting to extend  the analysis to more machines in the \emph{favorite} machines setting \cite{favorite}. There, an important parameter is also the minimum number $k$ of favorite machines per job. The case $k=1$ is perhaps interesting as, in the online setting, this gives a problem which is as difficult as the more general unrelated machines. Is it possible to characterize the $PoA$ and the $SPoA$ in this setting for any $s$? Do these bounds improve for larger $k$? Another interesting restriction would be the case of \emph{unit-size} jobs, which means that each job has processing time $1$ or $s$. Such \emph{two-values} restrictions have been studied in the \emph{mechanism design} setting with \emph{selfish machines} \cite{LavSwa09,Auletal15}, where players are machines and they possibly speculate on their true cost. Considering other well studied solution concepts would  also be interesting, including \emph{sequential} $PoA$ \cite{originalSPOA,bilo,SPOAcongestion,Gies17}, \emph{approximate} $SPoA$ \cite{feldman2009approximate}, and the  \emph{price of stochastic anarchy} \cite{ChuKatPhrRot08}. 

 \bibliographystyle{plain} 
 \bibliography{Mybib}

\begin{thebibliography}{10}

\bibitem{andelman2009strong}
Nir Andelman, Michal Feldman, and Yishay Mansour.
\newblock Strong price of anarchy.
\newblock {\em Games and Economic Behavior}, 2(65):289--317, 2009.

\bibitem{PoSoriginal}
Elliot Anshelevich, Anirban Dasgupta, Jon~M. Kleinberg, {\'{E}}va Tardos, Tom
  Wexler, and Tim Roughgarden.
\newblock The price of stability for network design with fair cost allocation.
\newblock {\em {SIAM} J. Comput.}, 38(4):1602--1623, 2008.

\bibitem{Auletal15}
Vincenzo Auletta, George Christodoulou, and Paolo Penna.
\newblock Mechanisms for scheduling with single-bit private values.
\newblock {\em Theory Comput. Syst.}, 57(3):523--548, 2015.

\bibitem{Aumann1959}
Robert~J. Aumann.
\newblock {\em Acceptable points in general cooperative n-person games}, pages
  287--324.
\newblock Princeton university press, 1959.

\bibitem{AWERBUCH2006200}
Baruch Awerbuch, Yossi Azar, Yossi Richter, and Dekel Tsur.
\newblock Tradeoffs in worst-case equilibria.
\newblock {\em Theoretical Computer Science}, 361(2):200 -- 209, 2006.

\bibitem{bilo}
Vittorio Bil{\`{o}}, Michele Flammini, Gianpiero Monaco, and Luca Moscardelli.
\newblock Some anomalies of farsighted strategic behavior.
\newblock {\em Theory Comput. Syst.}, 56(1):156--180, 2015.

\bibitem{favorite}
Cong Chen, Paolo Penna, and Yinfeng Xu.
\newblock Online scheduling of jobs with favorite machines.
\newblock submitted, 2017.

\bibitem{Chien2009}
Steve Chien and Alistair Sinclair.
\newblock {\em Strong and Pareto Price of Anarchy in Congestion Games}, pages
  279--291.
\newblock Springer Berlin Heidelberg, Berlin, Heidelberg, 2009.

\bibitem{christodoulou2005price}
George Christodoulou and Elias Koutsoupias.
\newblock The price of anarchy of finite congestion games.
\newblock In {\em Proceedings of the 37th annual ACM Symposium on Theory of
  Computing (STOC)}, pages 67--73, 2005.

\bibitem{ChuKatPhrRot08}
Christine Chung, Katrina Ligett, Kirk Pruhs, and Aaron Roth.
\newblock The price of stochastic anarchy.
\newblock In {\em Proc. of the 1st Int. Symp. on Algorithmic Game Theory
  (SAGT)}, volume 4997 of {\em LNCS}, pages 303--314, 2008.

\bibitem{czumaj2007tight}
Artur Czumaj and Berthold V{\"o}cking.
\newblock Tight bounds for worst-case equilibria.
\newblock {\em ACM Transactions on Algorithms (TALG)}, 3(1):4, 2007.

\bibitem{SPOAcongestion}
Jasper de~Jong and Marc Uetz.
\newblock The sequential price of anarchy for atomic congestion games.
\newblock In {\em Proc. of the 10th International Conference on Web and
  Internet Economics (WINE)}, volume 8877 of {\em LNCS}, pages 429--434, 2014.

\bibitem{Epstein2010}
Leah Epstein.
\newblock Equilibria for two parallel links: the strong price of anarchy versus
  the price of anarchy.
\newblock {\em Acta Informatica}, 47(7):375--389, 2010.

\bibitem{feldman2009approximate}
Michal Feldman and Tami Tamir.
\newblock Approximate strong equilibrium in job scheduling games.
\newblock {\em Journal of Artificial Intelligence Research}, 36:387--414, 2009.

\bibitem{feldmann2003nashification}
Rainer Feldmann, Martin Gairing, Thomas L{\"u}cking, Burkhard Monien, and
  Manuel Rode.
\newblock Nashification and the coordination ratio for a selfish routing game.
\newblock In {\em ICALP 2003}, volume~30, pages 514--526. Springer, 2003.

\bibitem{fiat2007strong}
Amos Fiat, Haim Kaplan, Meital Levy, and Svetlana Olonetsky.
\newblock Strong price of anarchy for machine load balancing.
\newblock In {\em ICALP 2007}, volume 4596, pages 583--594. Springer, 2007.

\bibitem{finn1979linear}
Greg Finn and Ellis Horowitz.
\newblock A linear time approximation algorithm for multiprocessor scheduling.
\newblock {\em BIT Numerical Mathematics}, 19(3):312--320, 1979.

\bibitem{gairing2006price}
Martin Gairing, Thomas L{\"u}cking, Marios Mavronicolas, and Burkhard Monien.
\newblock The price of anarchy for restricted parallel links.
\newblock {\em Parallel Processing Letters}, 16(01):117--131, 2006.

\bibitem{Gies17}
Paul Giessler, Akaki Mamageishvili, Mat{\'{u}}s Mihal{\'{a}}k, and Paolo Penna.
\newblock Sequential solutions in machine scheduling games.
\newblock {\em CoRR}, abs/1611.04159, 2016.

\bibitem{KlePilTar09}
Robert Kleinberg, Georgios Piliouras, and {\'E}va Tardos.
\newblock Multiplicative updates outperform generic no-regret learning in
  congestion games.
\newblock In {\em Proc. of the 41st Annual ACM Symp. on Theory of Computing
  (STOC)}, pages 533--542, 2009.

\bibitem{koutsoupias2003approximate}
Elias Koutsoupias, Marios Mavronicolas, and Paul Spirakis.
\newblock Approximate equilibria and ball fusion.
\newblock {\em Theory of Computing Systems}, 36(6):683--693, 2003.

\bibitem{koutsoupias1999worst}
Elias Koutsoupias and Christos Papadimitriou.
\newblock Worst-case equilibria.
\newblock In {\em Proceedings of the 16th Annual Symposium on Theoretical
  Aspects of Computer Science (STACS)}, pages 404--413. Springer-Verlag, 1999.

\bibitem{LavSwa09}
R.~Lavi and C.~Swamy.
\newblock Truthful mechanism design for multi-dimensional scheduling via cycle
  monotonicity.
\newblock {\em Games and Economic Behavior}, 67(1):99--124, 2009.

\bibitem{originalSPOA}
Renato~Paes Leme, Vasilis Syrgkanis, and {\'{E}}va Tardos.
\newblock The curse of simultaneity.
\newblock In {\em Proc. of Innovations in Theoretical Computer Science (ITCS)},
  pages 60--67, 2012.

\bibitem{roughgarden2009intrinsic}
Tim Roughgarden.
\newblock Intrinsic robustness of the price of anarchy.
\newblock In {\em Proc. of the 41st annual ACM Symposium on Theory of Computing
  (STOC)}, pages 513--522, 2009.

\bibitem{schuurman2007performance}
Petra Schuurman and Tjark Vredeveld.
\newblock Performance guarantees of local search for multiprocessor scheduling.
\newblock {\em INFORMS Journal on Computing}, 19(1):52--63, 2007.

\end{thebibliography}

\newpage

\begin{landscape}
	\centering
	\begin{longtable}{l|l|l|l|l|l}
	\caption{Subcases to prove Lemmas~\ref{lem.a1-1}-\ref{lem.c1-0c2-1} and Lemmas~\ref{cla.2}-\ref{cla.3}}\label{tab.UB}\\
	\hline

	\hline
	\scalebox{0.9}[1.0]{Lemma.subcases} & $s$ & \scalebox{0.8}[1.0]{$\xa{},\xaa{},\xc{},\xcc{}$} & Constrains needed & Weight coefficient & Bounds \\
	\hline
	\endfirsthead
	\caption[]{(continued)}\\
	\hline

	\hline
	\endhead
	\hline
	\endfoot
	\endlastfoot

	\ref{lem.a1-1}\ma 1 & $[1,\infty]$ & $1,\, 0,\, 0,\, \zeroone$
	& $\eqref{eq.ll1}$ & $\{1\}$ & $1$ \\

	\ref{lem.a1-1}\ma 2 & $[1,\infty]$ & $1,\, 0,\, 1,\, \zeroone$
	& $\eqref{eq.ll1}, \eqref{eq.ll2}, \eqref{eq.a1}, \eqref{eq.c1}$ & $ \frac{\{ {s^3+s}; {s^2+1}; {s^2}; {1}\}}{{s^3+s^2+1}}$ & $\frac{s^3+s^2+s+1}{s^3+s^2+1}$ \\

	\ref{lem.a1-1}\ma 3\mi a & $[1,\sqrt{2}]$ & $1,\, 1,\, 0,\, \zeroone$
	& $\eqref{eq.ll1}, \eqref{eq.ll2}, \eqref{eq.a1}, \eqref{eq.a2}$ & $\{ {2 s}; \,{2}; \,{s^2}; \,{2-s^2}\} \cdot \frac{1}{s+2}$ & $\frac{2 (s+1)}{s+2}$ \\

	\ref{lem.a1-1}\ma 3\mi b & $[\sqrt{2},\infty]$ & $1,\, 1,\, 0,\, \zeroone$
	& $\eqref{eq.ll1}, \eqref{eq.ll2}, \eqref{eq.l2}, \eqref{eq.a1}$ & $ \frac{\{ {s^2}; {s}; {s (s^2-2)}; {s (s^2-1)}\}}{s^3-s+1}$ & $\frac{s (s^2+s-1)}{s^3-s+1}$ \\

	\ref{lem.a1-1}\ma 3\mi c & $[1,\infty]$ & $1,\, 1,\, 0,\, \zeroone$
	& $\eqref{eq.ll2}, \eqref{eq.a1}, \eqref{eq.a2}$ & $\{ {2 s}; {s}; {s}\} \cdot \frac{1}{2 s-1}$ & $\frac{2 s}{2 s-1}$ \\

	\ref{lem.a1-1}\ma 4\mi a & $[1, 1.272]$ & $1,\, 1,\, 1,\, \zeroone$
	& $\eqref{eq.ll1}, \eqref{eq.ll2}, \eqref{eq.a1}, \eqref{eq.a2}, \eqref{eq.c1}$ & $ \frac{\{ {2 s^3+s}; {2 s^2+1}; {s^4}; {-s^4+s^2+1}; {s^2}\}}{s^3+2 s^2+s+1}$ & $\frac{2 s^3+2 s^2+s+1}{s^3+2 s^2+s+1}$ \\

	\ref{lem.a1-1}\ma 4\mi b & \scalebox{0.9}[1.0]{$[1.272,\infty]$} & $1,\, 1,\, 1,\, \zeroone$
	& $\eqref{eq.ll1}, \eqref{eq.ll2}, \eqref{eq.l2}, \eqref{eq.a1}, \eqref{eq.c1}$ & $ \frac{\{ {s^3}; {s^2}; {s^4-s^2-1}; {s^2 (s^2-1)}; {s^2-1}\}}{s^4+s-1}$ & $\frac{s^4+s^3-1}{s^4+s-1}$ \\ \hline

	\ref{lem.a2-0}\ma 1 & $[1,\infty]$ & $0,\, 0,\, 0,\, 1$
	& $\eqref{eq.ll1}$ & $\{1\}$ & $1$ \\

	\ref{lem.a2-0}\ma 2 & $[1,\infty]$ & $0,\, 0,\, 1,\, 0$
	& $\eqref{eq.ll2}$ & $\{1/s\}$ & $1/s$ \\

	\ref{lem.a2-0}\ma 3\mi a & $[1,\infty]$ & $0,\, 0,\, 1,\, 1$
	& $\eqref{eq.ll1}, \eqref{eq.ll2}, \eqref{eq.c1}, \eqref{eq.c2}$ & $\{ {2 s}; {2}; {1}; {1} \} \cdot \frac{1}{s+2}$ & $\frac{2 (s+1)}{s+2}$ \\

	\ref{lem.a2-0}\ma 3\mi b & $[1,\infty]$ & $0,\, 0,\, 1,\, 1$
	& $\eqref{eq.ll1}, \eqref{eq.ll2}, \eqref{eq.c1}$ & $\{ {s}; {2}; {1} \} \cdot \frac{1}{s+1}$ & $\frac{s+2}{s+1}$ \\ \hline

	\ref{lem.c1-1}\ma 1\mi a & $[1,\infty]$ & $0,\, 1,\, 1,\, 0$
	& $\eqref{eq.ll2}, \eqref{eq.l2}, \eqref{eq.c1}$ & $ \frac{\{ {s^2}; {s^2-1}; {s^2-1} \}}{s^2+s-1}$ & $\frac{2 s^2-1}{s^2+s-1}$ \\

	\ref{lem.c1-1}\ma 1\mi b\ma \eqref{eq.ab1} & \scalebox{0.9}[1.0]{$[1.272,\infty]$} & $0,\, 1,\, 1,\, 0$
	& $\eqref{eq.ll1}, \eqref{eq.ll2}, \eqref{eq.l2}, \eqref{eq.c1}, \eqref{eq.ab1}$ & $ \frac{\{ {s^3}; {s^4}; {s^4-s^2-1}; {s^4-1}; { s^4-s^2} \}}{{2 s^4-s^2+s-1}}$ & $\frac{2 s^4+s^3-s^2-1}{2 s^4-s^2+s-1}$ \\

	\ref{lem.c1-1}\ma 1\mi b\ma \eqref{eq.ab2} & \scalebox{0.9}[1.0]{$[1.272,\infty]$} & $0,\, 1,\, 1,\, 0$
	& $\eqref{eq.ll2}, \eqref{eq.a2}, \eqref{eq.ab2}$ & $ \frac{\{ {s^2-s+1}; {(s-1)^2 (s+1)}; {s (s^2-1)} \}}{{s^3-2 s^2+s+1}}$ & $\frac{s^2-s+1}{s^3-2 s^2+s+1}$ \\

		\ref{lem.c1-1}\ma 2\mi a\ma \eqref{eq.acbd1} & $[1,\infty]$ & $0,\, 1,\, 1,\, 1$
	& $\eqref{eq.l2}, \eqref{eq.a2}, \eqref{eq.c1}, \eqref{eq.acbd1}$ & $ \{  {s^2};{s^2-1};{1};{1} \}\cdot\frac{1}{s^2-s+1}$ & $\frac{s^2}{s^2-s+1}$ \\

		\ref{lem.c1-1}\ma 2\mi a\ma \eqref{eq.acbd2} & $[1,\infty]$ & $0,\, 1,\, 1,\, 1$
	& $\eqref{eq.l2}, \eqref{eq.a2}, \eqref{eq.c2}, \eqref{eq.acbd2}$ & $  \frac{\{ {s^2-s+2};{s^2};{1};{s} \}}{s^2-s+1}$ & $\frac{s^2-s+2}{s^2-s+1}$ \\

	\ref{lem.c1-1}\ma 2\mi b & $[1,\sqrt{2}]$ & $0,\, 1,\, 1,\, 1$
	& $\eqref{eq.ll1}, \eqref{eq.ll2}, \eqref{eq.a2}, \eqref{eq.c1}, \eqref{eq.c2}$ & $ \frac{\{  {s (s^2+2)}; {s^2+2}; {2-s^2}; {s^2}; {s^2} \}}{s^2+2 s+2}$ & $\frac{s^3+s^2+2 s+2}{s^2+2 s+2}$ \\

	\ref{lem.c1-1}\ma 2\mi c & $[\sqrt{2},\infty]$ & $0,\, 1,\, 1,\, 1$
	& $\eqref{eq.ll1}, \eqref{eq.ll2}, \eqref{eq.l2}, \eqref{eq.c1}$ & $ \frac{\{  {s}; {s^2}; {s^2-2}; {s^2-1} \}}{s^2+s-1}$ & $\frac{2 s^2+s-2}{s^2+s-1}$ \\

	\hline

	\ref{lem.c1-0c2-0}\mi a & $[1,\frac{1+\sqrt{5}}{2}]$ & $0,\, 1,\, 0,\, 0$
	& $\eqref{eq.ll2}$ & $\{s\}$ & $s$ \\

	\ref{lem.c1-0c2-0}\mi b\ma \eqref{eq.ab1}\mi a & $[1,\infty]$ & $0,\, 1,\, 0,\, 0$
	& $\eqref{eq.ll1}, \eqref{eq.ll2}, \eqref{eq.a2}, \eqref{eq.ab1}$ & $\{  {s}; {s^2}; {1}; {s^2-1} \} \cdot \frac{1}{s^2}$ & $1+\frac{1}{s}$ \\

	\ref{lem.c1-0c2-0}\mi b\ma \eqref{eq.ab1}\mi b\ma \eqref{eq.abd1}\mi a & $[1,\infty]$ & $0,\, 1,\, 0,\, 0$
	& $\eqref{eq.ll1}, \eqref{eq.ll2}, \eqref{eq.ab1}, \eqref{eq.abd1}$ & $ \frac{\{  {s^2}; {s (s^2-1)}; {s (s^2-1)}; {s} \}}{s^3-1}$ & $\frac{s (s^2+s-1)}{s^3-1}$ \\

	\scalebox{.9}[1]{\ref{lem.c1-0c2-0}\mi b\ma \eqref{eq.ab1}\mi b\ma \eqref{eq.abd1}\mi b\ma \eqref{eq.ad1}} & $[1,\infty]$ & $0,\, 1,\, 0,\, 0$
	& $\eqref{eq.ll2}, \eqref{eq.l2}, \eqref{eq.ab1}, \eqref{eq.ad1}$ & $ \{\frac{s}{2 s-1};\frac{s}{2 s-1};\frac{s}{2 s-1};\frac{s}{2 s-1} \}$ & $\frac{2 s}{2 s-1}$ \\

	\scalebox{.9}[1]{\ref{lem.c1-0c2-0}\mi b\ma \eqref{eq.ab1}\mi b\ma \eqref{eq.abd1}\mi b\ma \eqref{eq.ad2}} & $[1,\infty]$ & $0,\, 1,\, 0,\, 0$
	& $\eqref{eq.ll1}, \eqref{eq.ab1}, \eqref{eq.ad2}, \eqref{eq.abd1}$ & $ \frac{\{  {s^2};{s^2-s};{s^2-s};{s^2} \}}{2 s^2-3 s+1}$ & $\frac{s^2}{2 s^2-3 s+1}$ \\

	\ref{lem.c1-0c2-0}\mi b\ma \eqref{eq.ab1}\mi b\ma \eqref{eq.abd2} & $[1,\infty]$ & $0,\, 1,\, 0,\, 0$
	& $\eqref{eq.ll2}, \eqref{eq.a2}, \eqref{eq.abd2}$ & $\{  {s^2}; {s}; {s} \} \cdot \frac{1}{2 s-1}$ & $\frac{s^2}{2 s-1}$ \\

	\ref{lem.c1-0c2-0}\mi b\ma \eqref{eq.ab2} & $[1,\infty]$ & $0,\, 1,\, 0,\, 0$
	& -- & -- & -- \\ \hline

	\ref{lem.c1-0c2-1}\ma \eqref{eq.all1}\mi a & $[1,\infty]$ & $0,\, 1,\, 0,\, 1$
	& $\eqref{eq.ll1}, \eqref{eq.ll2}, \eqref{eq.a2}, \eqref{eq.all1}$ & $\frac{\{  {s (s+1)}; {s+1}; {s+1}; {s^2/(s-1)} \}}{{2 s+1}}$ & $\frac{(	s+1)^2}{2 s+1}$ \\

	\ref{lem.c1-0c2-1}\ma \eqref{eq.all1}\mi b & $[1,\infty]$ & $0,\, 1,\, 0,\, 1$
	& $\eqref{eq.l2}, \eqref{eq.a2}, \eqref{eq.all1}$ & $\{ \frac{s+1}{s}; \frac{s+1}{s}; \frac{1}{(s-1) s} \}$ & $1+\frac{1}{s}$ \\

	\ref{lem.c1-0c2-1}\ma \eqref{eq.all1}\mi c\ma \eqref{eq.ab1} & $[1,\infty]$ & $0,\, 1,\, 0,\, 1$
	& $\eqref{eq.ll1}, \eqref{eq.ll2}, \eqref{eq.a2}, \eqref{eq.c2}, \eqref{eq.ab1}$ & $ \frac{\{  {s^2+1}; {s^3+s}; {1/s}; {s}; {s^3-1/s} \}}{{s^3+s+1}}$ & $\frac{s^3+s^2+s+1}{s^3+s+1}$ \\

	\ref{lem.c1-0c2-1}\ma \eqref{eq.all1}\mi c\ma \eqref{eq.ab2}\mi a & $[1,\infty]$ & $0,\, 1,\, 0,\, 1$
	& $\eqref{eq.l2}, \eqref{eq.a2}, \eqref{eq.ab2}$ & $\{ \frac{s^2-s+1}{(s-1) s}; \frac{s}{s-1}; \frac{1}{s-1} \}$ & $\frac{s^2-s+1}{(s-1) s}$ \\

	\scalebox{0.9}[1]{\ref{lem.c1-0c2-1}\ma \eqref{eq.all1}\mi c\ma \eqref{eq.ab2}\mi b\ma \eqref{eq.abd1}} & $[1,\infty]$ & $0,\, 1,\, 0,\, 1$
	& $\eqref{eq.ll1}, \eqref{eq.abd1}, \eqref{eq.all1}$ & $\{  \frac{s+1}{2}; \frac{s+1}{2 s}; \frac{s^2+1}{2 (s-1) s} \}$ & $\frac{s+1}{2}$ \\

	\scalebox{.9}[1]{\ref{lem.c1-0c2-1}\ma \eqref{eq.all1}\mi c\ma \eqref{eq.ab2}\mi b\ma \eqref{eq.abd2}} & $[1,\infty]$ & $0,\, 1,\, 0,\, 1$
	& \scalebox{0.9}[1.0]{$\eqref{eq.ll1}, \eqref{eq.l2}, \eqref{eq.a2}, \eqref{eq.ab2}, \eqref{eq.abd2}$} & \scalebox{0.92}[1.0]{$\frac{\{  (s-1) (s^2+1); {s}; {s(s^2+1)}; {s^3+s-1}; {1} \}}{{s^3-s^2+s-1}}$} & $\frac{s^3-s^2+2s-1}{s^3-s^2+s-1}$ \\

	\ref{lem.c1-0c2-1}\ma \eqref{eq.all2}\mi a\ma \eqref{eq.ad1} & $[1,\infty]$ & $0,\, 1,\, 0,\, 1$
	& $\eqref{eq.ll1}, \eqref{eq.ll2}, \eqref{eq.all2}, \eqref{eq.ad1}$ & $\frac{\{  {s}; {s^2+s+1}; {s^2/(s-1)}; {s+1} \}}{{2 s+1}}$ & $\frac{(s+1)^2}{2 s+1}$ \\

	\ref{lem.c1-0c2-1}\ma \eqref{eq.all2}\mi a\ma \eqref{eq.ad2} & $[1,\infty]$ & $0,\, 1,\, 0,\, 1$
	& \scalebox{0.96}[1.0]{$\eqref{eq.ll1}, \eqref{eq.ll2}, \eqref{eq.a2}, \eqref{eq.c2}, \eqref{eq.ad2}$} & \scalebox{0.89}[1.0]{$\frac{\{{s^3+s}; {s^2+1}; {s^3+s^2-s}; {s+1- \frac{1}{s}}; {s^3- \frac{1}{s}} \}}{s^3+2}$} & $\frac{s^3+s^2+s+1}{s^3+2}$ \\

	\ref{lem.c1-0c2-1}\ma \eqref{eq.all2}\mi b\ma \eqref{eq.ab1} & $[1,\infty]$ & $0,\, 1,\, 0,\, 1$
	& $\eqref{eq.ll1}, \eqref{eq.ll2}, \eqref{eq.a2}, \eqref{eq.c2}, \eqref{eq.ab1}$ & $ \frac{\{  {s^2+1}; {s^3+s}; {1/s}; {s}; {s^3-1/s} \}}{s^3+s+1}$ & $\frac{s^3+s^2+s+1}{s^3+s+1}$ \\

	\ref{lem.c1-0c2-1}\ma \eqref{eq.all2}\mi b\ma \eqref{eq.ab2} & $[1,\infty]$ & $0,\, 1,\, 0,\, 1$
	& $\eqref{eq.l2},\, \eqref{eq.a2},\, \eqref{eq.ab2},\, \eqref{eq.all2}$ & $\frac{\{ {s^2};\, {s(s+1)};\, {s+1};\, {1/(s-1)} \}}{s^2-1}$ & $\frac{s^2}{s^2-1}$ \\

	\hline

	\ref{cla.2}\ma 1\mi a & $[1,\sqrt{2}]$ & $1,\, 1,\, \zeroone,\, 0$
	& $\eqref{eq.ll1}, \eqref{eq.ll2}, \eqref{eq.a1}, \eqref{eq.a2}$ & $\{  {2 s}; {2}; {s^2}; {2-s^2} \} \cdot \frac{1}{s+2}$ & $\frac{2 (s+1)}{s+2}$ \\

	\ref{cla.2}\ma 1\mi b & $[\sqrt{2},\infty]$ & $1,\, 1,\, \zeroone,\, 0$
	& $\eqref{eq.ll1}, \eqref{eq.ll2}, \eqref{eq.a1}$ & $\{ {s^2}; {s (s^2-1)}; {s} \} \cdot \frac{1}{s^2+s-1}$ & $s$ \\

	\ref{cla.2}\ma 2 & $[1,\infty]$ & $0,\, 1,\, 1,\, 0$
	& $\eqref{eq.ll1}, \eqref{eq.ll2}, \eqref{eq.c1}$ & $ \frac{\{ {s (s^2-1)}; {s^2}; {s^2-1} \}}{s^2+s-1}$ & $s$ \\

	\ref{cla.3} & $[1,\infty]$ & $\zeroone,\, 1,\, \zeroone,\, 1$
	& $\eqref{eq.ll1}, \eqref{eq.ll2}, \eqref{eq.a2}, \eqref{eq.c2}$ & $ \frac{\{ {s^3+s}; {s^2+1}; {1};\, {s^2} \}}{s^2+s+1}$ & $\frac{s^3+s^2+s+1}{s^2+s+1}$ \\
	\hline 

	\hline
	\multicolumn{6}{l}{\quad\footnotesize Note: ``$\zeroone$'' means \emph{either $0$ or $1$} in the third column.}
	\end{longtable}
\end{landscape}

\newpage

\appendix
\section{Postponed Proofs}

\subsection{Proofs of Lemmas~\ref{lem.a1-1}-\ref{lem.c1-0c2-1}}\label{app:postponed proofs:lemmas}
\begin{proof}[Proof of Lemma~\ref{lem.a1-1}]
	Since $\xa{} > 0$, we consider 
	four subcases:  Case~\ref{lem.a1-1}\ma{}1~($\xaa{}, \xc{}=0$), Case~\ref{lem.a1-1}\ma{}2~($\xaa{}=0, \xc{}>0$), Case~\ref{lem.a1-1}\ma{}3~($\xaa{}>0, \xc{}=0$) and Case~\ref{lem.a1-1}\ma{}4~($\xaa{}, \xc{}>0)$, which correspond to rows \ref{lem.a1-1}\ma{}1 to \ref{lem.a1-1}\ma{}4\mi{}b in Table~\ref{tab.UB}.
	
	The bound for Case~\ref{lem.a1-1}\ma{}3 is the minimum of the bounds of \ref{lem.a1-1}\ma{}3\mi{}a, \ref{lem.a1-1}\ma{}3\mi{}b and \ref{lem.a1-1}\ma{}3\mi{}c. Similarly, the bound for Case~\ref{lem.a1-1}\ma{}4 is the minimum of the bounds of \ref{lem.a1-1}\ma{}4\mi{}a and \ref{lem.a1-1}\ma{}4\mi{}b.
	Finally, the maximum of the bounds of the four subcases give the bound for this lemma.
	These 4 subcases are summed up in Figure~\ref{fig:a1-1} (orange line).
\end{proof}

\begin{proof}[Proof of Lemma~\ref{lem.a2-0}]
	Since $\xa{} = \xaa{} = 0$, we consider three subcases: Case \ref{lem.a2-0}\ma{}1 ($\xc{}=0,\xcc{}>0$), Case~\ref{lem.a2-0}\ma{}2~($\xc{}>0,\xcc{}=0$) and Case~\ref{lem.a2-0}\ma{}3~($\xc{}>0,\xcc{}>0$), which correspond to rows \ref{lem.a2-0}\ma{}1 to \ref{lem.a2-0}\ma{}3\mi{}b in Table~\ref{tab.UB}.

	The bound for Case~\ref{lem.a2-0}\ma{}3 is the minimum of the bounds of \ref{lem.a2-0}\ma{}3\mi{}a and \ref{lem.a2-0}\ma{}3\mi{}b.
	Finally, the maximum of the bounds of the three subcases give the bound for this lemma.
	These 3 subcases are summed up in Figure~\ref{fig:a2-0} (orange line).
\end{proof}

\begin{proof}[Proof of Lemma~\ref{lem.c1-1}]
	Since $\xa{} = 0,\, \xaa{},\xc{}>0$, we consider two subcases: Case~\ref{lem.c1-1}\ma{}1 ($\xcc{}=0$) and Case \ref{lem.c1-1}\ma{}2~($\xcc{}>0$), which correspond to rows \ref{lem.c1-1}\ma{}1\mi{}a to \ref{lem.c1-1}\ma{}2\mi{}c in Table~\ref{tab.UB}.
	
	The bound for Case~\ref{lem.c1-1}\ma{}1 is the minimum of the bounds of \ref{lem.c1-1}\ma{}1\mi{}a and \ref{lem.c1-1}\ma{}1\mi{}b, where \ref{lem.c1-1}\ma{}1\mi{}b is the maximum of the bounds of \ref{lem.c1-1}\ma{}1\mi{}b\ma{}\eqref{eq.ab1} and \ref{lem.c1-1}\ma{}1\mi{}b\ma{}\eqref{eq.ab2}. Note that it only needs one of constraints \eqref{eq.ab1} and \eqref{eq.ab2} holds, thus we take the maximum of \ref{lem.c1-1}\ma{}1\mi{}b\ma{}\eqref{eq.ab1} and \ref{lem.c1-1}\ma{}1\mi{}b\ma{}\eqref{eq.ab2}.
	The bound for Case~\ref{lem.c1-1}\ma{}2 is the minimum of the bounds of \ref{lem.c1-1}\ma{}2\mi{}a, \ref{lem.c1-1}\ma{}2\mi{}b, \ref{lem.c1-1}\ma{}2\mi{}c, where \ref{lem.c1-1}\ma{}2\mi{}a is the maximum of the bounds of \ref{lem.c1-1}\ma{}2\mi{}a\ma{}\eqref{eq.acbd1} and \ref{lem.c1-1}\ma{}2\mi{}a\ma{}\eqref{eq.acbd2}.
	Finally, the maximum of the bounds of the two subcases give the bound for this lemma.
	These 2 subcases are summed up in Figure~\ref{fig:c1-1}.
\end{proof}

\begin{proof}[Proof of Lemma~\ref{lem.c1-0c2-0}]
	This lemma considers the case $\xa{} = \xc{} = \xcc{} = 0$ and $\xaa{} > 0$, which correspond to rows \ref{lem.c1-0c2-0}\mi{}a to \ref{lem.c1-0c2-0}\mi{}b\ma{}\eqref{eq.ab2} in Table~\ref{tab.UB}.
	
	Similar as the above proofs, the bound for this case is the minimum of the bounds of \ref{lem.c1-0c2-0}\mi{}a and \ref{lem.c1-0c2-0}\mi{}b, where \ref{lem.c1-0c2-0}\mi{}b is the maximum of the bounds of \ref{lem.c1-0c2-0}\mi{}b\ma{}\eqref{eq.ab1} and \ref{lem.c1-0c2-0}\mi{}b\ma{}\eqref{eq.ab2}. 
	Furthermore, the bound of \ref{lem.c1-0c2-0}\mi{}b\ma{}\eqref{eq.ab1} is the minimum of the bounds of \ref{lem.c1-0c2-0}\mi{}b\ma{}\eqref{eq.ab1}\mi{}a and \ref{lem.c1-0c2-0}\mi{}b\ma{}\eqref{eq.ab1}\mi{}b, where \ref{lem.c1-0c2-0}\mi{}b\ma{}\eqref{eq.ab1}\mi{}b is the maximum of \ref{lem.c1-0c2-0}\mi{}b\ma{}\eqref{eq.ab1}\mi{}b\ma{}\eqref{eq.abd1} and \ref{lem.c1-0c2-0}\mi{}b\ma{}\eqref{eq.ab1}\mi{}b\ma{}\eqref{eq.abd2}. Moreover \ref{lem.c1-0c2-0}\mi{}b\ma{}\eqref{eq.ab1}\mi{}b\ma{}\eqref{eq.abd1} is the minimum of \ref{lem.c1-0c2-0}\mi{}b\ma{}\eqref{eq.ab1}\mi{}b\ma{}\eqref{eq.abd1}\mi{}a and \ref{lem.c1-0c2-0}\mi{}b\ma{}\eqref{eq.ab1}\mi{}b\ma{}\eqref{eq.abd1}\mi{}b, where \ref{lem.c1-0c2-0}\mi{}b\ma{}\eqref{eq.ab1}\mi{}b\ma{}\eqref{eq.abd1}\mi{}b is the maximum of \ref{lem.c1-0c2-0}\mi{}b\ma{}\eqref{eq.ab1}\mi{}b\ma{}\eqref{eq.abd1}\mi{}b\ma{}\eqref{eq.ad1} and \ref{lem.c1-0c2-0}\mi{}b\ma{}\eqref{eq.ab1}\mi{}b\ma{}\eqref{eq.abd1}\mi{}b\ma{}\eqref{eq.ad2}.
	Finally, the bounds give the bound for this lemma, which is shown in Figure~\ref{fig:c1-0c2-0}.
\end{proof}

\begin{proof}[Proof of Lemma~\ref{lem.c1-0c2-1}]
	This lemma consider the case $\xa{} = \xc{} = 0$ and $\xaa{}, \xcc{} > 0$, which correspond to rows \ref{lem.c1-0c2-1}\ma{}\eqref{eq.all1}\mi{}a to \ref{lem.c1-0c2-1}\ma{}\eqref{eq.all2}\mi{}b\ma{}\eqref{eq.ab2} in Table~\ref{tab.UB}.
	
	Similar as the above proofs, the bound for this case is the maximum of the bounds of \ref{lem.c1-0c2-1}\ma{}\eqref{eq.all1} and \ref{lem.c1-0c2-1}\ma{}\eqref{eq.all2}, where \ref{lem.c1-0c2-1}\ma{}\eqref{eq.all1} is the minimum of \ref{lem.c1-0c2-1}\ma{}\eqref{eq.all1}\mi{}a, \ref{lem.c1-0c2-1}\ma{}\eqref{eq.all1}\mi{}b and \ref{lem.c1-0c2-1}\ma{}\eqref{eq.all1}\mi{}c, and \ref{lem.c1-0c2-1}\ma{}\eqref{eq.all2} is the minimum of \ref{lem.c1-0c2-1}\ma{}\eqref{eq.all2}\mi{}a and \ref{lem.c1-0c2-1}\ma{}\eqref{eq.all2}\mi{}b.
	Furthermore, the bound of \ref{lem.c1-0c2-1}\ma{}\eqref{eq.all1}\mi{}c is the maximum of the bounds of \ref{lem.c1-0c2-1}\ma{}\eqref{eq.all1}\mi{}c\ma{}\eqref{eq.ab1} and \ref{lem.c1-0c2-1}\ma{}\eqref{eq.all1}\mi{}c\ma{}\eqref{eq.ab2}, where \ref{lem.c1-0c2-1}\ma{}\eqref{eq.all1}\mi{}c\ma{}\eqref{eq.ab2} is the minimum of \ref{lem.c1-0c2-1}\ma{}\eqref{eq.all1}\mi{}c\ma{}\eqref{eq.ab2}\mi{}a and \ref{lem.c1-0c2-1}\ma{}\eqref{eq.all1}\mi{}c\ma{}\eqref{eq.ab2}\mi{}b where \ref{lem.c1-0c2-1}\ma{}\eqref{eq.all1}\mi{}c\ma{}\eqref{eq.ab2}\mi{}b is the maximum of \ref{lem.c1-0c2-1}\ma{}\eqref{eq.all1}\mi{}c\ma{}\eqref{eq.ab2}\mi{}b\ma{}\eqref{eq.abd1} and \ref{lem.c1-0c2-1}\ma{}\eqref{eq.all1}\mi{}c\ma{}\eqref{eq.ab2}\mi{}b\ma{}\eqref{eq.abd2}.
	Besides,
	the bound of \ref{lem.c1-0c2-1}\ma{}\eqref{eq.all2}\mi{}a is the maximum of bounds of \ref{lem.c1-0c2-1}\ma{}\eqref{eq.all2}\mi{}a\ma{}\eqref{eq.ad1} and \ref{lem.c1-0c2-1}\ma{}\eqref{eq.all2}\mi{}a\ma{}\eqref{eq.ad2}, and \ref{lem.c1-0c2-1}\ma{}\eqref{eq.all2}\mi{}b is the maximum of \ref{lem.c1-0c2-1}\ma{}\eqref{eq.all2}\mi{}b\ma{}\eqref{eq.ab1} and \ref{lem.c1-0c2-1}\ma{}\eqref{eq.all2}\mi{}b\ma{}\eqref{eq.ab2}.
	Finally, the bounds give the bound for this lemma, which is shown in Figure~\ref{fig:c1-0c2-1}. 
\end{proof}

\subsection{Proof of lower bound (Lemma~\ref{lem.LB})} 
\label{app:proof_of_lower_bound}

Table~\ref{tab.LU} gives instances that match all the bounds of Theorem~\ref{thm.UB}.
Each instance is represented by the setting of Figure~\ref{fig:setting} and each symbol (e.g. $\xa{}$, $\xaa{}$, $\xb{}\dots$) represents at most one job.

\begin{proposition}
	The optimal makespan of each instance of Table~\ref{tab.LU} is at most 1.
\end{proposition}

\begin{proof}
	It is easy to see that, for every instance of Table~\ref{tab.LU}, the optimum schedule shown in Figure~\ref{fig:setting} satisfies $\max\{\ell_1^*, \ell_2^*\} \le 1$.
	Hence we can have $opt \le 1$ even though the above schedule is not guaranteed to be a SE,
	because the optimal SE can only have better makespan than that schedule.
\end{proof}

\begin{proposition}\label{prop.LBNE}
	Each schedule of Table~\ref{tab.LU} is a NE.
\end{proposition}
\begin{proof}
	Because $\ell_1 \ge \ell_2$, no job in machine 2 can reduce its cost by moving to machine 1.
	Thus we only need to check that if there is no single job in machine 1 would benefit from moving to machine 2.
	\begin{description}
		\item[LB1] If $\xaa{}$ moves to machine 2, $\ell_2' = \ell_2 + \xaa{}/s = \frac{s^3+s^2+s+1}{s^3+2} = \ell_1$.

		If $\xcc{}$ moves to machine 2, $\ell_2' = \ell_2 + \xcc{} \cdot s = \frac{s^3+s^2+s+1}{s^3+2} = \ell_1$.

		\item[LB2] If $\xaa{}$ moves to machine 2, $\ell_2' = \ell_2 + \xaa{}/s = \frac{s^2+2s+1}{2s+1} = \ell_1$.

		If $\xcc{}$ moves to machine 2, $\ell_2' = \ell_2 + \xcc{} \cdot s = \frac{2s^2+2s}{2s+1} \ge \ell_1 = \frac{s^2+2s+1}{2s+1}$.

		\item[LB3] If $\xaa{}$ moves to machine 2, $\ell_2' = \ell_2 + \xaa{}/s = \frac{s+1}{s} = \ell_1$.

		If $\xcc{}$ moves to machine 2, $\ell_2' = \ell_2 + \xcc{} \cdot s = 2 \ge \ell_1 = \frac{s+1}{s}$.

		\item[LB4] If $\xaa{}$ moves to machine 2, $\ell_2' = \ell_2 + \xaa{}/s = \frac{s^3-s^2+2s-1}{s^3-s^2+s-1} = \ell_1$.

		If $\xcc{}$ moves to machine 2, $\ell_2' = \ell_2 + \xcc{} \cdot s = \frac{s^4-s^3+s^2-1}{s^3-s^2+s-1} \ge \ell_1 = \frac{s^3-s^2+2s-1}{s^3-s^2+s-1}$ by $s \ge s_3 \approx 1.755$.

		\item[LB5] If $\xaa{}$ moves to machine 2, $\ell_2' = \ell_2 + \xaa{}/s = \frac{3}{2} \ge \ell_1 = \frac{s+1}{2}$ by $s \le s_5 =2$.

		If $\xcc{}$ moves to machine 2, $\ell_2' = \ell_2 + \xcc{} \cdot s = \frac{s+2}{2} \ge \ell_1 = \frac{s+1}{2}$.

		\item[LB6] If $\xaa{}$ moves to machine 2, $\ell_2' = \ell_2 + \xaa{}/s = \frac{s^2-s+1}{s^2-s} = \ell_1$.

		If $\xcc{}$ moves to machine 2, $\ell_2' = \ell_2 + \xcc{} \cdot s = s \ge \ell_1 = \frac{s^2-s+1}{s^2-s}$ by $s \ge s_5 = 2$.

		\item[LB7] If $\xaa{}$ moves to machine 2, $\ell_2' = \ell_2 + \xaa{}/s = \frac{s^2}{2s-1} = \ell_1$.

		\item[LB8] If $\xaa{}$ moves to machine 2, $\ell_2' = \ell_2 + \xaa{}/s = 2 \ge \ell_1 = \frac{s+1}{s}$.
	\end{description}
\end{proof}
According to Proposition~\ref{prop.LBNE} we know that no single job can reduce its cost by moving to other machine.
It also holds that

\begin{proposition}[Epstein \cite{Epstein2010}]\label{prop.swap}
Given a schedule on two machines which is a NE, if this schedule is not a SE, then a coalition of jobs where every job can reduce its cost consists of at least one job of each one of the machines.
\end{proposition}

\begin{proposition}\label{prop.a2in}
	If a schedule of Table~\ref{tab.LU} is not a SE, then a coalition of jobs where every job can reduce its cost must consist of job $\xaa{}$.
\end{proposition}

\begin{proof}
	Because good job will become bad job after swap, if the coalition consists of only good jobs, it holds that $\ell_1'+\ell_2' \ge \ell_1+\ell_2$.
	Thus at least one of the cost of the jobs will get worse after swap.
	Therefore, if a schedule is not a SE, then a coalition of jobs where every job can reduce its cost must consist of at least one bad job.

	For LB7 and LB8, $\xaa{}$ is in the coalition for sure, since $\xaa{}$ is the only good job.
	For other instances in Table~\ref{tab.LU}, if $\xaa{}$ is not in the coalition then $\xbb{}$ must in it according to Proposition~\ref{prop.swap}.
	Next we will show that $\xaa{}+\xbb{}/s \ge \ell_2$ for every instance of them, so that if $\xbb{}$ is in the coalition then $\xbb{}$ will not benefit from moving to machine 1. Thus $\xaa{}$ must be in the coalition.
	For LB1 to LB3, we have $\xaa{} \ge \ell_2$ thus $\xaa{}+\xbb{}/s \ge \ell_2$ holds.
	For LB4, $\xaa{}+\xbb{}/s = \frac{2s^2-s}{s^3-s^2+s-1} \ge \ell_2 = 1$ by $s \le s_4$.
	For LB5, $\xaa{}+\xbb{}/s = \frac{s}{2} +\frac{1}{2} \ge \ell_2 = 1$.
	For LB6, $\xaa{}+\xbb{}/s = \frac{1}{s-1} +\frac{1}{s} \ge \ell_2 = 1$ by $s \le s_6 \approx 2.154$.	
\end{proof}

\begin{proposition}\label{prop.a2only}
	For any instance of Table~\ref{tab.LU}, if $\xaa{}$ is the only job of machine 1 that in a coalition of jobs, these jobs cannot reduce their costs at the same time.
\end{proposition}
\begin{proof}
	According to Proposition~\ref{prop.swap}, if $\xaa{}$ is the only job of machine 1 that in a coalition of jobs, there must be some job of machine 2 in this coalition. Thus this proposition says there are no such a $\xaa{}\text{-}\{*\}$ swap that all these jobs can benefit simultaneously, where $\{*\}$ is any subset of jobs of machine 2.
	We prove this for the 8 instances of Table~\ref{tab.LU} one by one.
	\begin{description}
	\item[LB1] In this case $\xaa{}$ moves to machine 2 and $\xcc{}$ stays, it holds that no job (or subset of jobs) in machine 2 can benefit from moving to machine 1, because
	\begin{description}
	 	\item[No $\xaa \text{-} \xbb{}$ swap:] $\ell_1' = \xcc{} + \xbb{}/s = \frac{s+2}{s^3+2} \ge \ell_2 = \frac{s^3+1}{s^3+2}$ by $s \le s_1 \approx 1.325$;

	 	\item[No $\xaa \text{-} \xd{}$ swap:] $\ell_1' = \xcc{} + \xd{} \cdot s = \frac{s^3+1}{s^3+2} = \ell_2$;

	 	\item[No $\xaa \text{-} \xdd{}$ swap:] $\ell_1' = \xcc{} + \xdd{} \cdot s = \frac{s^4-s^3-s^2+3s+1}{s^3+2} > \ell_2 = \frac{s^3+1}{s^3+2}$ by $s \le s_1$.
	\end{description}

	\item[LB2] In this case $\xaa{}$ moves to machine 2 and $\xcc{}$ stays, it holds that no job (or subset of jobs) in machine 2 can benefit from moving to machine 1, because
	\begin{description}
	 	\item[No $\xaa \text{-} \xbb{}$ swap:] $\ell_1' = \xcc{} + \xbb{}/s = 1 \ge \ell_2 = \frac{s^2+s}{2s+1}$ by $s \le s_2 \approx 1.618$;
	 	\item[No $\xaa \text{-} \xdd{}$ swap:] $\ell_1' = \xcc{} + \xdd{} \cdot s = \frac{s^2+s+1}{2s+1} > \ell_2 = \frac{s^2+s}{2s+1}$.
	\end{description}

	\item[LB3] In this case $\xaa{}$ moves to machine 2 and $\xcc{}$ stays, it holds that no job (or subset of jobs) in machine 2 can benefit from moving to machine 1, because
	\begin{description}
	 	\item[No $\xaa \text{-} \xbb{}$ swap:] $\ell_1' = \xcc{} + \xbb{}/s = 1 = \ell_2$;
	 	\item[No $\xaa \text{-} \xdd{}$ swap:] $\ell_1' = \xcc{} + \xdd{} \cdot s = \frac{1}{s}+2s-s^2 > \ell_2 = 1$ by $s \le s_3 \approx  1.755$.
	\end{description}

	\item[LB4] In this case $\xaa{}$ moves to machine 2 and $\xcc{}$ stays, it holds that no job (or subset of jobs) in machine 2 can benefit from moving to machine 1, because
	\begin{description}
	 	\item[No $\xaa \text{-} \xbb{}$ swap:] $\ell_1' = \xcc{} + \xbb{}/s = 1 = \ell_2$;
	 	\item[No $\xaa \text{-} \xdd{}$ swap:] $\ell_1' = \xcc{} + \xdd{} \cdot s = 1 = \ell_2$.
	\end{description}

	\item[LB5] In this case $\xaa{}$ moves to machine 2 and $\xcc{}$ stay, it holds that no job (or subset of jobs) in machine 2 can benefit from moving to machine 1, because
	\begin{description}
	 	\item[No $\xaa \text{-} \xbb{}$ swap:] $\ell_1' = \xcc{} + \xbb{}/s = 1 = \ell_2$;
	 	\item[No $\xaa \text{-} \xdd{}$ swap:] $\ell_2' = \xbb{} + \xaa{} /s = \ell_1 = \frac{s+1}{2} $;
	 	\item[No $\xaa \text{-} \{ \xbb{},\xdd{}\}$ swap:] $\ell_1' = \xcc{} + \xbb{}/s + \xdd{} \cdot s = \frac{-s^2+2s+2}{2} \ge \ell_2 = 1$.
	\end{description}

	\item[LB6] In this case $\xaa{}$ moves to machine 2 and $\xcc{}$ stays, it holds that no job (or subset of jobs) in machine 2 can benefit from moving to machine 1, because
	\begin{description}
	 	\item[No $\xaa \text{-} \xbb{}$ swap:] $\ell_1' = \xcc{} + \xbb{}/s = 1 = \ell_2$.
	\end{description}

	\item[LB7] In this case $\xaa{}$ moves to machine 2, it holds that no job (or subset of jobs) in machine 2 can benefit from moving to machine 1, because
	\begin{description}
	 	\item[No $\xaa \text{-} \xd{}$ swap:] $\ell_1' = \xd{} \cdot s = \frac{s(s-1)^2}{2s-1} \ge \ell_2 = \frac{s^2-s}{2s-1}$ by $s \ge s_6 \approx 2.154$;
	 	\item[No $\xaa \text{-} \xdd{}$ swap:] $\ell_1' = \xdd{} \cdot s = \frac{s^2-s}{2s-1} = \ell_2$.
	\end{description}

	\item[LB8] In this case $\xaa{}$ moves to machine 2, it holds that no job (or subset of jobs) in machine 2 can benefit from moving to machine 1, because
	\begin{description}
	 	\item[No $\xaa \text{-} \xd{}$ swap:] $\ell_1' = \xd{} \cdot s = 1 \ge \ell_2 = \frac{s^2-1}{s^2}$;
	 	\item[No $\xaa \text{-} \xdd{}$ swap:] $\ell_1' = \xdd{} \cdot s = \frac{s^2-s-1}{s} \ge \ell_2 = \frac{s^2-1}{s^2}$ by $s \ge s_7 \approx 2.247$.
	\end{description}
	\end{description}
\end{proof}

\begin{lemma}
	Each schedule of Table~\ref{tab.LU} is a SE.
\end{lemma}

\begin{proof}
According to Propositions~\ref{prop.a2in} and \ref{prop.a2only}, we know that schedules of LB7 and LB8 are SE, since $\xaa{}$ is the only job in machine 1.
For LB1 to LB6, we will prove there is no such coalition of jobs can reduce their costs at the same time:
\begin{description}
	\item[LB1] According to Proposition~\ref{prop.a2in}, we only need to consider the case both $\xaa{}$ and $\xcc{}$ of machine 1 are in the coalition.
	However, we have
	\begin{description}
		\item[No $\{\xaa{},\xcc{}\}\text{-}\{\xbb{},\xdd{}\}$ swap:] $\ell_2' = \xd{} + \xaa{}/s + \xcc{} \cdot s = \frac{3s^2+2s-1}{s^3+2} \ge \ell_1 = \frac{s^3+s^2+s+1}{s^3+2}$ by $s \le s_1$.
	\end{description}
	Since $\frac{s^2-1}{s^3+2} \le \frac{s^3-s^2-s+2}{s^3+2} \le \frac{s}{s^3+2}$ by $s \le s_1$, i.e., $\xd{}\le \xdd{} \le \xbb{}$, we can know that $\xaa{}$ and $\xcc{}$ will not benefit even if the smallest job $\xd{}$ of machine 2 stays still.
	Thus all jobs in machine 2 should be in the coalition. Since
	\begin{description}
		\item[No $\{\xaa{},\xcc{}\}\text{-}\{\xbb{},\xd{},\xdd{}\}$ swap:] $\ell_1' = \xbb{}/s + \xd{} \cdot s + \xdd{} \cdot s = \frac{s^4 - s^2 +s +1}{s^3+2} \ge \ell_2 = \frac{s^3+1}{s^3+2}$ by $s \le s_1$,
	\end{description}
	we know that there is no such coalition of jobs that every job can reduce its cost by deviating simultaneously.

	\item[LB2] Similar as the former case, we consider the case both $\xaa{}$ and $\xcc{}$ are in the coalition.
	We get the same result by the fact that
	\begin{description}
		\item[No $\{\xaa{},\xcc{}\}\text{-}\{*\}$ swap:] $\ell_2' \ge \xaa{}/s + \xcc{} \cdot s = \frac{s^2+2s+1}{2s+1} = \ell_1$.
	\end{description}

	\item[LB3] Considering the case both $\xaa{}$ and $\xcc{}$ are in the coalition, we have that
	\begin{description}
		\item[No $\{\xaa{},\xcc{}\}\text{-}\{*\}$ swap:] $\ell_2' \ge \xaa{}/s + \xcc{} \cdot s = \frac{s+1}{s} = \ell_1$.
	\end{description}

	\item[LB4] Considering the case both $\xaa{}$ and $\xcc{}$ are in the coalition, we have that
	\begin{description}
		\item[No $\{\xaa{},\xcc{}\}\text{-}\{*\}$ swap:] $\ell_2' \ge \xaa{}/s + \xcc{} \cdot s = \frac{s^4-ss^3+2s^2}{s^3-s^2+s-1} \ge \ell_1 = \frac{s^3-s^2+2s-1}{s^3-s^2+s-1}$ by $s \ge s_3 \approx 1.755$.
	\end{description}

	\item[LB5] Considering the case both $\xaa{}$ and $\xcc{}$ are in the coalition, we have that
	\begin{description}
		\item[No $\{\xaa{},\xcc{}\}\text{-}\{*\}$ swap:] $\ell_2' \ge \xaa{}/s + \xcc{} \cdot s = \frac{s+1}{2} = \ell_1$.
	\end{description}

	\item[LB6] Considering the case both $\xaa{}$ and $\xcc{}$ are in the coalition, we have that
	\begin{description}
		\item[No $\{\xaa{},\xcc{}\}\text{-}\{*\}$ swap:] $\ell_2' \ge \xaa{}/s + \xcc{} \cdot s = \frac{1}{s^2-s}+s-1 \ge \ell_1 = \frac{s^2-s+1}{s^2-s}$ by $s \ge s_5 = 2$.
	\end{description}
\end{description}
\end{proof}

\subsection{Proofs of Lemmas~\ref{cla.1}-\ref{cla.4}} 
\label{app:proofs_of_PoA}
\begin{proof}[Proof of Lemma~\ref{cla.1}]
	Rows \ref{lem.a1-1}\ma 1, \ref{lem.a1-1}\ma 2 and \ref{lem.a2-0}\ma 4\mi a in Table~\ref{tab.UB} cover all the cases. Since each of the 3 bounds is below the claimed bound $\frac{s^3+s^2+s+1}{s^2+s+1}$, the lemma is proved.
\end{proof}
\begin{proof}[Proof of Lemma~\ref{cla.2}]
	Similar to the proof of SPoA, according to rows \ref{cla.2}\ma{}1\mi{}a to \ref{cla.2}\ma{}2 of Table~\ref{tab.UB} we have 
	\[
		\max\Big\{\min\big\{\frac{2(s+1)}{s+2},s \big\},s \Big\} \le \frac{s^3+s^2+s+1}{s^2+s+1}
	\]
	for $s \ge 1$, thus the lemma is proved.
\end{proof}
\begin{proof}[Proof of Lemma~\ref{cla.3}]
	The bound of row \ref{cla.3} of Table~\ref{tab.UB} implies this lemma.
\end{proof}
\begin{proof}[Proof of Lemma~\ref{cla.4}]
	We know that
	\begin{gather*}
	\ell_1 = \xaa{} +\xcc{} = \frac{s^3+s^2+s+1}{s^2+s+1}\,,\\
	\ell_2 = \xb{} +\xbb{} = \frac{s^3+1}{s^2+s+1} \,.
	\end{gather*}
	If $\xaa{}$ goes to machine 2, the cost of $\xaa{}$ becomes
	\[
		\ell_2' = \ell_2 + \xaa{}/s = \frac{s^3+s^2+s+1}{s^2+s+1}\,.
	\]
	Similarly, if $\xcc{}$ goes to machine 2, the cost of $\xcc{}$ becomes
	\[
		\ell_2' = \ell_2 + \xcc{}\cdot s = \frac{s^3+s^2+s+1}{s^2+s+1}\,.
	\]
	Thus the schedule is a NE.
\end{proof}

\end{document}